\documentclass{article}

\usepackage{arxiv}

\usepackage[utf8]{inputenc} 
\usepackage[T1]{fontenc}    
\usepackage{hyperref}       
\usepackage{url}            
\usepackage{booktabs}       
\usepackage{array}
\usepackage{amsfonts}       
\usepackage{nicefrac}       
\usepackage{graphicx}
\usepackage{verbatim}
\usepackage{subcaption}
\usepackage{adjustbox}
\usepackage[numbers,sort&compress]{natbib}
\usepackage{doi}
\usepackage{amsmath}
\usepackage{bm}
\usepackage{amssymb}
\usepackage{float}
\usepackage[section]{placeins} 
\usepackage{etoolbox}
\pretocmd{\subsection}{\FloatBarrier}{}{} 
\usepackage{xcolor}
\usepackage{mathtools}
\usepackage{fancyhdr}
\usepackage{cleveref}
\crefname{appendix}{Appendix}{Appendices}
\Crefname{appendix}{Appendix}{Appendices}
\usepackage{amsmath,amssymb,amsthm}
\usepackage{bm}
\usepackage{mathtools}
\usepackage{stmaryrd}
\usepackage{enumitem}
\usepackage{booktabs}
\usepackage{makecell}
\usepackage{graphicx}
\usepackage{subcaption}
\usepackage[mathscr]{euscript}
\usepackage{textcomp}

\DeclareMathAlphabet{\pazocal}{OMS}{zplm}{m}{n}

\newcommand{\tempvar}{1.0}

\newcommand{\jump}[1]{\llbracket #1 \rrbracket}
\newcommand{\dd}{\,\mathrm{d}}
\newcommand{\R}{\mathbb{R}}
\newcommand{\subdiff}{\partial}

\newcommand{\pos}[1]{\langle #1 \rangle_{+}}

\usepackage{lscape}		
\usepackage{fancyvrb}

\newtheoremstyle{remarkstyle}  
  {5pt}                        
  {5pt}                        
  {}                           
  {}                           
  {\bfseries}                  
  {.}                          
  { }                          
  {}                           
\newtheorem{theorem}{Theorem}[section]

\newtheorem{lemma}[theorem]{Lemma}

\theoremstyle{remarkstyle}
\newtheorem{remark}[theorem]{Remark}

\crefname{figure}{Fig.}{Figs.}
\Crefname{figure}{Figure}{Figures}
\crefname{algorithm}{Algorithm}{Algorithms}
\crefname{equation}{Eq.}{Eqs.}
\crefname{table}{Table}{Tables}
\Crefname{table}{Table}{Tables}
\crefname{section}{Section}{Sections}
\crefname{remark}{Remark}{Remarks}
\Crefname{remark}{Remark}{Remarks}

\usepackage[normalem]{ulem}

\title{
Data-Driven Cohesive Zone Modeling within the Generalized Standard Materials Framework
}

\author{
Sida Hao$^{1}$, 
Jinkyo Han$^{1}$, 
Bahador Bahmani$^{1,2}$\thanks{Corresponding author: \texttt{bahador.bahmani@northwestern.edu}} \\
\\
$^{1}$Department of Mechanical Engineering, Northwestern University, Evanston, IL 60208, USA \\
$^{2}$Theoretical and Applied Mechanics, Northwestern University, Evanston, IL 60208, USA
}

\renewcommand{\shorttitle}{Neural Generalized Standard Cohesive Models}

\hypersetup{
    colorlinks=true,      
    linkcolor=green,      
    urlcolor=blue,        
    citecolor=blue        
}

\date{}

\begin{document}
\maketitle

\vspace{-10pt}

\begin{abstract}
Cohesive zone models are widely used to describe fracture and interfacial failure, yet most formulations prescribe problem-specific analytical traction–separation laws together with phenomenological rules for unloading and reloading, leading to specialized models for different cohesive behaviors. This work develops a unified learnable cohesive formulation within the generalized standard materials framework, in which the response is generated from learned constitutive functions while the underlying thermodynamic structure remains fixed. The surface free energy is decomposed into active and contact contributions, and irreversible damage evolution is governed by a learned mode-dependent damage resistance. The active energy is represented by an input-convex neural network, while the inverse damage resistance is represented by a monotone neural network. Convexity, monotonicity, normalization, and damage irreversibility are incorporated directly into the constitutive representation. Direct parameterization of the inverse resistance yields an explicit damage update and avoids local nonlinear inversion during constitutive evaluation. Material-point studies show that the formulation can represent qualitatively distinct cohesive responses, including plateaus, extended softening tails, irregular softening, nonlinear unloading, distinct Mode I and Mode II behaviors, and several classical mixed-mode cohesive laws. The framework therefore replaces law-specific model construction with a single thermodynamically structured representation capable of learning cohesive responses of broad functional complexity from data.

\end{abstract}

\keywords{
Cohesive Zone Model
\and
Interfacial Damage Model
\and
Traction--Separation Law
\and
Generalized Standard Material
\and
Scientific Machine Learning
}

\section{Introduction}
\label{sec:intro}

Fracture in engineering materials often involves a finite process zone and cannot be described solely by the loss of cohesion across an ideally sharp crack. Within this zone, separation may involve mechanisms such as microcracking, plastic deformation, fiber bridging, interfacial debonding, or adhesive degradation. Cohesive zone models (CZMs) provide a reduced description of these mechanisms by replacing the detailed process-zone physics with a traction--separation relation acting across a displacement jump between separating surfaces \cite{Dugdale1960,Barenblatt1962,hillerborg1976analysis,needleman1987continuum,tvergaard1992relation,xu1994,camacho1996computational,ortiz1999finite}.

CZMs are widely used to model fracture and interfacial failure when the structure and dissipation of the process zone influence the macroscopic response. Many formulations prescribe the traction--separation relation through fixed analytical forms \citep{park2011}, including bilinear \citep{hillerborg1976analysis,geubelle1998impact}, piecewise-linear \citep{de2021piecewise}, and trapezoidal \citep{tvergaard1992relation} laws. These phenomenological models are simple and effective, but their prescribed functional forms can restrict the range of process-zone responses that can be represented in complex interfaces and composite materials \citep{wei2026}.

A further challenge arises in mixed-mode fracture, where normal and tangential separations evolve jointly and their coupling must be represented consistently in the cohesive response \citep{Park2009,mcgarry2014potential,dimitri2015coupled}. Potential-based formulations provide a systematic way to introduce this coupling by deriving the traction vector from a scalar interface potential. For a sufficiently smooth potential, the resulting traction field is integrable and satisfies reciprocal normal--tangential coupling. Early examples include the polynomial potential of \citet{needleman1987continuum} and the exponential potential of \citet{xu1993void}, later used by \citet{xu1994} to model dynamic crack growth and branching. The Park--Paulino--Roesler (PPR) model subsequently introduced separate control of the Mode~I and Mode~II fracture energies, cohesive strengths, and shape parameters \citep{Park2009,Park2012,park2011}.

Despite these advantages, the range of cohesive responses represented by a
potential-based CZM remains limited by the prescribed form of the potential.
Classical formulations typically employ low-parameter analytical functions,
which may be restrictive for experimentally identified cohesive responses
exhibiting plateaus \citep{sorensen2003determination}, extended softening tails
\citep{sorensen1998large}, or irregular and strongly mode-dependent shapes
\citep{ansari2025rotation,sorensen2003determination}. This limitation is particularly important under large-scale bridging, where
the process zone is not small relative to the structural dimensions and the
response can depend on the detailed shape of the cohesive law rather than only
on its peak strength and fracture energy
\citep{bao1992remarks,sorensen1998large}.

A separate limitation concerns irreversible evolution. An interface potential
that depends only on the current separation describes a reversible response
and retains no memory of prior decohesion. It therefore cannot by itself
represent damage accumulation, hysteresis, or irreversible fracture
dissipation. Classical CZMs introduce such behavior through damage or history
variables together with prescribed unloading--reloading rules
\citep{yang1998single,nguyen2001cohesive,hordijk1991growth,kuna2015general}.
A thermodynamically consistent formulation instead requires internal variables
and evolution laws that distinguish recoverable interface energy from
irreversible dissipation and satisfy the dissipation inequality
\citep{mosler2011thermodynamically,dimitri2015coupled}. Moreover, the existence
of an interface potential alone does not guarantee a physically admissible
mixed-mode response; particular potential forms can produce residual or
repulsive tractions and other nonphysical behavior
\citep{mcgarry2014potential,dimitri2015coupled}. These considerations motivate
a constitutive structure in which the stored energy, internal variables, and
irreversible evolution are treated within a common thermodynamic framework.

The generalized standard material (GSM) framework \citep{halphen1975materiaux} provides a systematic thermodynamic structure by combining a Helmholtz
free-energy potential for the recoverable response with a convex dissipation potential governing the evolution of internal variables. Under the standard
GSM assumptions, this construction ensures non-negative dissipation and provides evolution laws for irreversible internal variables, while the
free-energy potential determines the recoverable response at fixed internal state. The framework has been extended to cohesive interfaces in adhesion--friction models \citep{raous1999consistent}, rate-dependent formulations \citep{Leuschner2015}, and rate-independent delamination models \citep{Roubicek2013,Vodicka2017}. Related thermodynamic interface models have also been formulated through damage surfaces and consistency conditions \citep{willam2004interface}. These approaches provide a consistent basis for irreversible cohesive fracture, but their constitutive potentials and evolution functions are generally prescribed analytically. Their representational capacity therefore remains tied to the assumed functional
forms.

Recent advances in data-driven constitutive modeling provide a route to increase constitutive flexibility while retaining physical structure. Since the early use of neural networks for learning material response directly from data \citep{ghaboussi1991}, the field has progressed from black-box constitutive approximations toward models that incorporate physical principles either through the training objective or directly through the model architecture \citep{fuhg2025review}. Thermodynamics-informed approaches impose conditions such as non-negative dissipation through physics-based loss terms \citep{vlassis2021sobolev}, whereas structure-preserving architectures can enforce properties such as convexity, monotonicity, and polyconvexity by construction \citep{as2022mechanics,klein2022polyconvex}. In particular, convex neural networks and neural GSM formulations have been used to learn free-energy and dissipation potentials for a range of bulk constitutive models \citep{thakolkaran2022nn,Flaschel2025,Xu2025,AmiriHezaveh2025,Friedrichs2026}.

Neural constitutive modeling of interfaces has received comparatively less
attention. Existing approaches include automated discovery of
traction--separation laws through deep reinforcement learning
\citep{wang2019meta}, thermodynamics-informed neural representations of
mixed-mode traction--separation behavior \citep{wei2023data}, and neural
cohesive laws for monotonic and cyclic fracture
\citep{tao2022neural,tao2024reconstruction}. Related developments have also
introduced learned potential representations for rate-and-state friction
\citep{liu2025learning}. Despite these advances, a general data-driven
cohesive formulation that combines flexible representation of the local
traction--separation response with the thermodynamic structure of GSM remains
lacking. In particular, thermodynamic restrictions introduced through penalty
terms in the training objective require balancing additional loss
contributions and do not hold exactly by construction. This motivates a
formulation in which the required constitutive properties are embedded
directly in the neural parametrization of the free-energy and dissipation
mechanisms.

To address this gap, this paper develops a thermodynamically structured,
data-driven cohesive zone model within the GSM framework. Rather than
prescribing the traction--separation law and unloading rules directly, the
cohesive response is generated from learned constitutive functions. The active
potential is represented by an input-convex neural network (ICNN)
\citep{amos2017input}, constructed to satisfy the required convexity and
monotonicity conditions, while the inverse damage resistance is represented by
a monotone neural network (MNN) \citep{kim2024scalable}. The latter directly
parametrizes the mode-mixity-dependent inverse resistance and therefore avoids
the numerical inversion of the resistance relation during constitutive
evaluation. We consider both an MNN with unrestricted dependence on mode
mixity and a full-MNN that additionally imposes increasing resistance with mode
mixity, motivated by interfaces whose fracture toughness increases from
Mode~I toward Mode~II. In contrast to approaches that impose thermodynamic
conditions through additional terms in the training objective
\citep{wei2023data}, the required convexity, monotonicity, and normalization
properties are embedded directly in the constitutive parametrization.

The resulting formulation provides an explicit irreversible damage update
without a local nonlinear solve for the damage state. The constitutive update
is differentiable almost everywhere, which permits the neural constitutive
functions to be identified directly from traction--separation data by
gradient-based optimization through the complete loading history.

Numerical experiments demonstrate that the formulation can reproduce a broad range of cohesive responses, including bilinear, trapezoidal, exponential, half-sine, irregular softening, and PPR laws, while retaining the imposed
thermodynamic structure. The examples also show that a single identified model can represent distinct Mode~I and Mode~II responses and predict the response at intermediate mode mixities not included in the calibration data.

The remainder of the paper is organized as follows.
\Cref{sec:gpb_czm} develops the cohesive formulation within the GSM framework,
and \cref{sec:nn_representation} introduces the ICNN and MNN parametrizations
of the active potential and inverse damage resistance.
\Cref{sec:training} defines the calibration dataset and identification
objective.
\Cref{sec:examples} presents the numerical examples for Mode~I and mixed-mode
cohesive responses, including a comparison of the MNN and full-MNN
representations.
\Cref{sec:conclusions} summarizes the main findings.
The appendices collect the convex-analysis background (\cref{app:convex}), 
the convexity results used in the formulation
(\cref{app:composed-convexity}), implementation and hyperparameter details
(\cref{app:hyperparams}), additional mixed-mode results
(\cref{app:mixed-mode-result}), and the reference cohesive laws used in the
classical-model study (\cref{app:dimitri_laws}).

\section{A GSM-Based Cohesive Zone Model}
\label{sec:gpb_czm}

This section first briefly reviews the fundamental ingredients of cohesive zone models and then presents a flexible formulation of cohesive interfaces as an interfacial damage model within the GSM framework. This formulation naturally decomposes the constitutive description into three independent functions, each of which can be parametrized arbitrarily, provided that fundamental thermodynamic requirements --- such as convexity, monotonicity, and other admissibility constraints --- are satisfied.

\subsection{Cohesive Zone Models}
\label{sec:czm_intro}

A cohesive zone model represents the fracture process zone as an internal
interface $\Gamma$ that transmits a cohesive traction
$\boldsymbol{t}$. Let $\boldsymbol{u}^{+}$ and $\boldsymbol{u}^{-}$ denote the
displacement traces on the two sides of the interface, as shown in \cref{fig:czm_diagram}(a). The displacement jump,
or separation vector, is defined by
\begin{equation}
  \boldsymbol{\delta}
  \coloneqq
  \jump{\boldsymbol{u}}
  \coloneqq
  \boldsymbol{u}^{+}-\boldsymbol{u}^{-}.
\end{equation}
A cohesive zone model can be specified by a traction--separation law
$\boldsymbol{t}(\boldsymbol{\delta})$. The traction typically increases to a
cohesive strength and subsequently decreases to zero, at which point a
traction-free crack forms, as illustrated in \cref{fig:czm_diagram}(b). In one
dimension, the area under the traction--separation curve equals the fracture
energy $G_c$.

In two dimensions,
\begin{equation}
  \boldsymbol{\delta}
  =
  \delta_n\boldsymbol{e}_n+\delta_t\boldsymbol{e}_t,
  \qquad
  \delta_n=\boldsymbol{\delta}\cdot\boldsymbol{e}_n,
  \qquad
  \delta_t=\boldsymbol{\delta}\cdot\boldsymbol{e}_t,
\end{equation}
where $\boldsymbol{e}_n$ and $\boldsymbol{e}_t$ are the unit normal and tangent
to the interface. Here, $\delta_n>0$ denotes opening, while $\delta_t$ measures
sliding and its sign specifies the sliding direction. 
The traction is decomposed similarly as
\begin{equation}
  \boldsymbol{t}=t_n\boldsymbol{e}_n+t_t\boldsymbol{e}_t,
  \qquad
  t_n=\boldsymbol{t}\cdot\boldsymbol{e}_n,
  \qquad
  t_t=\boldsymbol{t}\cdot\boldsymbol{e}_t.
\end{equation}
Pure opening, defined by $\delta_n>0$ and $\delta_t=0$, is termed Mode~I,
whereas pure sliding, defined by $\pos{\delta_n}=0$ and $\delta_t\neq0$, is
termed Mode~II. General separation is mixed-mode. Following
\citet{turon2006damage}, the mode mixity is quantified by
\begin{equation}
  \beta(\boldsymbol{\delta})
  =
  \frac{\lvert\delta_t\rvert}
       {\pos{\delta_n}+\lvert\delta_t\rvert}
  \in[0,1],
  \qquad
  \pos{\delta_n}\coloneqq \max(\delta_n,0),
  \label{eq:mode-mixity}
\end{equation}
with $\beta=0$ for Mode~I and $\beta=1$ for Mode~II. By convention,
$\beta=0$ when the denominator vanishes.

\begin{figure}[htbp]
  \centering
  \begin{subfigure}[b]{0.32\linewidth}
    \centering
    \includegraphics[width=0.99\linewidth]{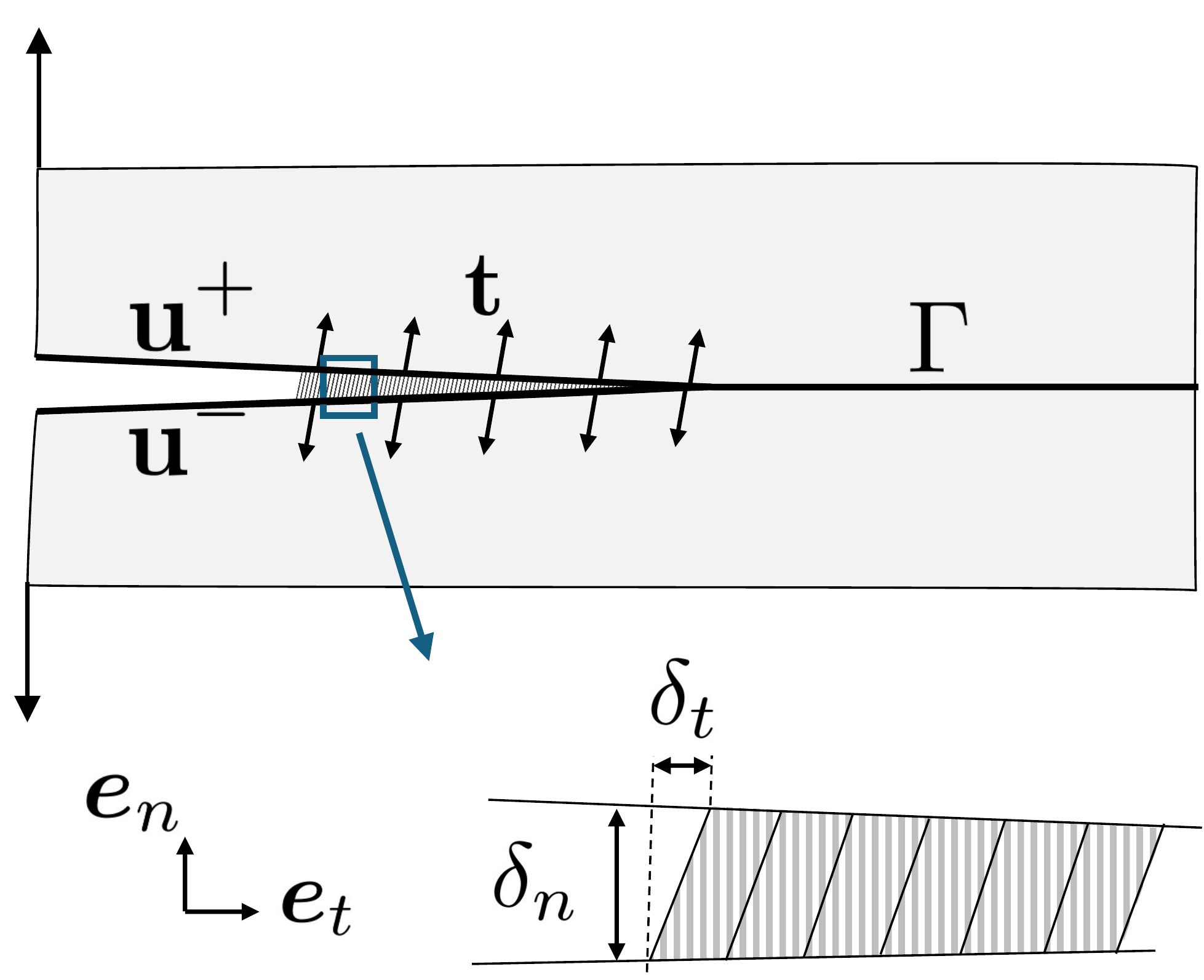}
    \caption{\label{fig:czm_diagram_crack} Opening crack and process zone}
  \end{subfigure}
  \hspace{0.03\linewidth}
  \begin{subfigure}[b]{0.32\linewidth}
    \centering
    \includegraphics[width=0.9\linewidth]{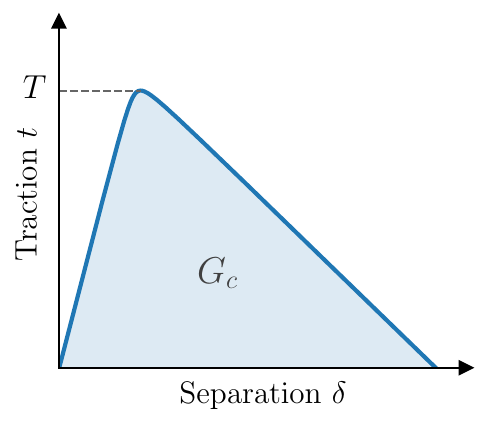}
    \caption{\label{fig:czm_diagram_tsl} Traction--separation law}
  \end{subfigure}
  \caption{Cohesive zone model of an opening crack. (a) A process zone ahead of the traction-free crack bridges the two faces with a traction $\boldsymbol{t}$ resisting their separation $\boldsymbol{\delta}$; the interface itself is denoted $\Gamma$. (b) The behavior is governed by a traction--separation law $\boldsymbol{t}(\boldsymbol{\delta})$: the traction rises to a cohesive strength $T$ and then softens to zero, with the area under the curve equal to the fracture energy $G_c$. The schematic shows a 1D model for clarity. In general, $\boldsymbol{t}$ and $\boldsymbol{\delta}$ are vectors.}
  \label{fig:czm_diagram}
\end{figure}

The remainder of this section formulates the traction--separation law within
the generalized standard materials framework, which ensures thermodynamic
admissibility by construction.

\subsection{A Generalized Standard Formulation}

This section formulates the proposed rate-independent cohesive law within the
generalized standard materials framework
\citep{halphen1975materiaux,raous1999consistent}. The free energy determines
the state laws, while a convex dissipation potential and the associated
normality relation govern damage evolution. The model is specified by an
active energy, a contact energy, and a scalar damage-resistance function.
Relevant concepts from convex analysis are summarized in \cref{app:convex}.

All energetic quantities and powers are measured per unit reference area of
the interface $\Gamma$. The local state is described by the separation
$\boldsymbol{\delta}$ and a scalar damage variable
\begin{equation}
  d\in[0,1],
  \qquad
  d=0 \ \text{(intact)},
  \qquad
  d=1 \ \text{(fully decohered)}.
\end{equation}

\subsubsection{Surface Free Energy}
\label{sec:surf-free-energ}
Within the GSM framework, the constitutive response is derived from a surface free energy $\psi(\boldsymbol{\delta},d)$. The free surface energy is required to be non-negative, convex, and continuously
differentiable with respect to $\boldsymbol{\delta}$ for every fixed
$d\in[0,1]$.

To distinguish the interface mechanisms associated with separation and
contact, we adopt the tension--compression decomposition \citep{kikuchi1988contact,willam2004interface}
\begin{equation}
  \psi(\boldsymbol{\delta},d)
  =
  (1-d)\,\psi^{e}_{+}(\boldsymbol{\delta}_{+})
  +
  \psi^{e}_{-}(\pos{-\delta_n}),
  \qquad
  \boldsymbol{\delta}_{+}
  \coloneqq
  \pos{\delta_n}\boldsymbol{e}_n+\delta_t\boldsymbol{e}_t.
  \label{eq:psi}
\end{equation}
The active energy $\psi^{e}_{+}$ represents the recoverable energy associated
with opening and sliding and is degraded by damage. The contact energy
$\psi^{e}_{-}$ represents the compressive normal penalty response and is not degraded.
This decomposition is a modeling assumption of the present work. It permits
the cohesive and contact responses to be treated separately and prevents
damage from degrading the compressive normal response.

A sufficient structural condition for the total free energy in
\cref{eq:psi} to satisfy the stated properties is that the two composed
contributions, regarded as functions of $\boldsymbol{\delta}$, are
individually non-negative, convex, and continuously differentiable. We further
impose the normalization
\begin{equation}
  \psi^{e}_{+}(\boldsymbol{0})=0,
  \qquad
  \psi^{e}_{-}(0)=0.
  \label{eq:energy-normalization}
\end{equation}
Since $1-d\ge0$, these conditions ensure that $\psi(\cdot,d)$ is
non-negative, convex, and continuously differentiable for every fixed
$d\in[0,1]$. 

The interface traction and the damage-driving force follow from the free
energy as
\begin{align}
  \boldsymbol{t}
  &\coloneqq
  \frac{\partial \psi}{\partial \boldsymbol{\delta}}
  =
  (1-d)
  \frac{\partial}{\partial \boldsymbol{\delta}}
  \left[
    \psi^{e}_{+}(\boldsymbol{\delta}_{+})
  \right]
  +
  \frac{\partial}{\partial \boldsymbol{\delta}}
  \left[
    \psi^{e}_{-}(\pos{-\delta_n})
  \right],
  \label{eq:traction_state_law}
  \\
  Y
  &\coloneqq
  -\frac{\partial \psi}{\partial d}
  =
  \psi^{e}_{+}(\boldsymbol{\delta}_{+}).
  \label{eq:damage_driving_force}
\end{align}
Here, $\boldsymbol{t}$ is the traction transmitted across the interface, and
$Y$ is the thermodynamic force conjugate to damage. Under pure normal
compression, $\delta_n\le0$ and $\delta_t=0$, one has
$\boldsymbol{\delta}_{+}=\boldsymbol{0}$ and hence $Y=0$. Closed sliding may still drive damage because the tangential component
remains active.
At $d=1$, the active opening and shear tractions vanish, whereas the contact
response remains active.

Cohesive laws commonly assume symmetry with respect to reversal of the
tangential separation
\citep{xu1994,willam2004interface,van2006improved,Park2009,dimitri2015coupled}.
Under this assumption, the active energy can be expressed as
\begin{equation}
  \psi^{e}_{+}
  =
  \psi^{e}_{+}\bigl(\pos{\delta_n},\lvert\delta_t\rvert\bigr).
\end{equation}
The corresponding composed active energy is convex in
$\boldsymbol{\delta}$ if and only if $\psi^{e}_{+}$ is jointly convex and
non-decreasing in both arguments; see
\cref{lem:positive-magnitude-composition}. Consequently, for
$\delta_t\neq0$,
\begin{equation*}
  t_t
  =
  (1-d)
  \frac{\partial\psi^{e}_{+}}
       {\partial\lvert\delta_t\rvert}
  \frac{\delta_t}{\lvert\delta_t\rvert}.
\end{equation*}
Since $\psi^{e}_{+}$ is non-decreasing in $\lvert\delta_t\rvert$, the
tangential traction has the same sign as $\delta_t$. This structure precludes
a reversal of the tangential traction relative to the tangential separation,
a nonphysical response that can arise in some potential-based cohesive laws
under mixed-mode loading \citep{mcgarry2014potential}.

If the two tangential directions are instead distinguished, the active energy
takes the form
\begin{equation}
  \psi^{e}_{+}
  =
  \psi^{e}_{+}\bigl(\pos{\delta_n},\delta_t\bigr).
\end{equation}
In this case, the corresponding composed energy is convex in
$\boldsymbol{\delta}$ if and only if $\psi^{e}_{+}$ is jointly convex and
non-decreasing in its first argument; see \cref{lem:composed-convexity}. No monotonicity condition is required
with respect to the signed tangential argument $\delta_t$.

Similarly, the contact contribution
$\psi^{e}_{-}(\pos{-\delta_n})$ is convex in $\boldsymbol{\delta}$ if and only
if $\psi^{e}_{-}$ is convex and non-decreasing in its argument.

\subsubsection{Dissipation Potential}
\label{sec:dissip-poten}

We next introduce the dissipation potential
$\phi(\,\cdot\,,d,\beta):\R\to\R\cup\{+\infty\}$, where $\beta$ is the
current mode mixity defined in \cref{eq:mode-mixity}. For each fixed
$(d,\beta)$, the function $\phi$ is assumed to be proper, non-negative,
convex, and lower semicontinuous in $\dot d$, with
$\phi(0,d,\beta)=0$. Rate independence is imposed through positive
homogeneity of degree one,
\begin{equation}
  \phi(c\dot d,d,\beta)
  =
  c\,\phi(\dot d,d,\beta),
  \qquad c>0.
  \label{eq:homog}
\end{equation}

Because $d$ is restricted to $[0,1]$, the admissible damage rates depend on
the current state. Irreversibility requires $\dot d\geq0$ for $d<1$, while
saturation at complete decohesion requires $\dot d=0$ at $d=1$. We therefore
define
\begin{equation}
  \mathcal{K}(d)
  =
  \begin{cases}
    [0,\infty), & 0\leq d<1,\\
    \{0\},      & d=1.
  \end{cases}
  \label{eq:rate-cone}
\end{equation}
The dissipation potential is finite only for $\dot d\in \mathcal{K}(d)$. For a scalar
damage variable, convexity and positive homogeneity of degree one then imply
that $\phi$ is linear in $\dot d$ on the admissible branch.

For $d<1$, convexity and positive homogeneity of degree one imply that the
dissipation potential is linear in $\dot d$ on the admissible half-line
$[0,\infty)$. Including the constraint $\dot d\in \mathcal{K}(d)$, it takes the form
\begin{equation}
  \phi(\dot d,d,\beta)
  =
  \begin{cases}
    r(d,\beta)\,\dot d, & \dot d\in \mathcal{K}(d),\\
    +\infty,            & \dot d\notin \mathcal{K}(d),
  \end{cases}
  \qquad
  r(d,\beta)\geq0,
  \label{eq:phi-reduced}
\end{equation}
where $r(d,\beta)$ is the damage resistance. The function $r$ specifies the
rate-independent dissipation and its dependence on the mode mixity.

The generalized normality relation \citep{halphen1975materiaux,raous1999consistent} gives 
\begin{equation}
  Y
  \in
  \subdiff_{\dot d}\phi(\dot d,d,\beta).
  \label{eq:normality}
\end{equation}
Since $\phi$ is convex and positively homogeneous of degree one,
\cref{eq:normality} implies
\begin{equation}
  Y\dot d
  =
  \phi(\dot d,d,\beta)
  \ge0.
\end{equation}
Together with the state relations in \cref{eq:damage_driving_force,eq:traction_state_law}, this yields the
local isothermal dissipation inequality
\begin{equation}
  \mathcal{D}
  =
  \boldsymbol{t}\cdot\dot{\boldsymbol{\delta}}
  -
  \dot\psi
  =
  Y\dot d
  \ge0.
  \label{eq:dissipation}
\end{equation}

For $d<1$, the normality relation is equivalent to the KKT conditions
\begin{equation}
  f(Y,d,\beta)
  \coloneqq
  Y-r(d,\beta)
  \le0,
  \qquad
  \dot d\ge0,
  \qquad
  \dot d\,f(Y,d,\beta)=0.
  \label{eq:KKT}
\end{equation}
Hence, damage remains fixed while $Y<r(d,\beta)$ and evolves only when
$Y=r(d,\beta)$. The activation function $f$ is therefore derived from the
dissipation potential rather than prescribed independently. At $d=1$, the
constraint $\mathcal{K}(1)=\{0\}$ prevents further damage evolution.

Because the affine degradation law in \cref{eq:psi} makes
$Y=\psi^{e}_{+}(\boldsymbol{\delta}_{+})$ independent of $d$, the KKT
conditions admit an explicit update when $r(\cdot,\beta)$ is continuous and
strictly increasing. Define
\begin{equation}
  d_{\mathrm{tr}}(t)
  =
  \begin{cases}
    0,
    & Y(t)\le r(0,\beta(t)),\\
    R\!\bigl(Y(t),\beta(t)\bigr),
    & r(0,\beta(t))<Y(t)<r(1,\beta(t)),\\
    1,
    & Y(t)\ge r(1,\beta(t)),
  \end{cases}
  \label{eq:trial-damage}
\end{equation}
where $R(Y,\beta) \coloneqq r^{-1}(Y,\beta)$ denotes the inverse of $r$ with respect to $d$, i.e., $R(r(d,\beta),\beta)=d$. The damage variable then evolves according to
\begin{equation}
  d(t)
  =
  \max_{\tau \le t} d_{\mathrm{tr}}(\tau).
  \label{eq:damage_update}
\end{equation}

The running maximum guarantees damage irreversibility, while the definition of $d_{\mathrm{tr}}$ ensures that $0\le d\le1$. Damage initiates when the driving force reaches $Y_0(\beta)=r(0,\beta)$. The strict monotonicity of $r$ with respect to $d$ guarantees that the inverse $R$ in \cref{eq:trial-damage} is single-valued. If $r$ were only non-decreasing, several values of $d$ could correspond to the same value of $Y$, so the inverse would not be uniquely defined. In addition, a forward parameterization of $r$ would generally require solving $r(d,\beta)=Y$ for $d$ during each constitutive update. This additional local nonlinear solve would increase the computational cost of inference and complicate training because gradients must be propagated through the solver.

For complete decohesion at fixed mode mixity, the dissipated energy per unit
reference area is defiend as,
\begin{equation}
  G_c(\beta)
  \coloneqq
  \int_{\text{process}}Y\dot d\,\dd t
  =
  \int_0^1 r(d,\beta)\,\dd d.
  \label{eq:Gc}
\end{equation}
The integral of $r$ therefore determines the mode-dependent fracture energy,
with $G_{Ic}=G_c(0)$ and $G_{IIc}=G_c(1)$, while the variation of $r$ with
$d$ determines how this energy is dissipated during damage growth. If
$\beta$ varies during the process, the dissipated energy becomes
$\int_0^1 r(d,\beta(d))\,\dd d$ and is generally path-dependent.

\begin{remark}
Thermodynamic admissibility requires $r\ge0$. Continuity and strict monotonicity of $r$ with respect to $d$ are additional modeling assumptions that ensure a well-defined, single-valued inverse $R$. These assumptions permit the explicit damage update above and avoid an additional local nonlinear solve.
\end{remark}

\subsection{Specialization of the Constitutive Functions}
\label{sec:specialization}

The constitutive model is defined by three functions: the active energy
$\psi^{e}_{+}$, the contact energy $\psi^{e}_{-}$, and the damage resistance
$r(d,\beta)$. Their specific forms are introduced below.

\paragraph{Contact potential.}
The compressive contact response is represented using the standard quadratic penalty potential \citep{kikuchi1988contact},
\begin{equation}
  \psi^{e}_{-}\bigl(\pos{-\delta_n}\bigr)
  =
  \frac{1}{2}K_{nc}\pos{-\delta_n}^{2},
  \qquad
  K_{nc}>0,
  \label{eq:psi-minus}
\end{equation}
where $\pos{-\delta_n}=\max(-\delta_n,0)$ is the penetration measure and
$K_{nc}$ is the normal penalty stiffness. The potential is non-negative,
convex, non-decreasing in its argument, and continuously differentiable, with
zero value and slope at $\delta_n=0$. A sufficiently large $K_{nc}$ limits
negative normal separation and provides a penalty approximation of the
impenetrability constraint.

\paragraph{Active potential.}
We consider the shear-symmetric case discussed in \cref{sec:surf-free-energ} and represent the active potential $\psi^{e}_{+}$ by an input-convex neural network (ICNN) \citep{amos2017input} with inputs $(\pos{\delta_n},\lvert\delta_t\rvert)$. The network is constrained to be jointly convex and non-decreasing in both inputs. The tangent normalization introduced in \cref{sec:nn_representation:icnn} preserves these properties while imposing zero energy and zero traction at the undeformed state. In the numerical examples, we use $\psi^{e}\equiv\psi^{e}_{+}$ for brevity.

\paragraph{Damage resistance.}
For the reasons discussed in \cref{sec:dissip-poten}, we directly parameterize the inverse resistance $R(Y,\beta)=r^{-1}(Y,\beta)$ and require it to be strictly increasing in $Y$. This condition is equivalent to strict monotonicity of $r(d,\beta)$ with respect to $d$. We represent $R$ using the partially monotone neural network of \citet{kim2024scalable}, with monotonicity imposed with respect to $Y$ and unrestricted dependence on $\beta$. We refer to this architecture as the MNN. For interfaces whose fracture toughness increases with mode mixity
\citep{cao1989experimental,Hutchinson1991,liechti1992asymmetric,
ansari2025rotation}, we additionally consider a fully monotone variant, denoted as the full-MNN, in which $R(Y,\beta)$ is constrained to be increasing in $Y$ and decreasing in $\beta$. The latter condition is equivalent to requiring $r(d,\beta)$ to increase with $\beta$ at fixed damage. This additional monotonicity is introduced as an empirically motivated inductive bias rather than a general constitutive requirement. It is particularly useful when the calibration data sparsely sample the mode-mixity space.

\begin{remark}
\label[remark]{rem:tangential-3d}
In three dimensions, the tangential separation
$\boldsymbol{\delta}_t$ is a vector in the interface plane. Tangential symmetry
then generalizes to \emph{in-plane isotropy}, whereby the potentials depend on
$\boldsymbol{\delta}_t$ only through its magnitude
$\lVert\boldsymbol{\delta}_t\rVert$ \citep{Leuschner2015}.
The composition result in \cref{lem:composed-convexity} extends directly to
this vector-valued setting. Under this assumption, the model does not distinguish between Mode~II and Mode~III and cannot represent in-plane shear anisotropy or non-coaxial tangential traction--separation responses; see \cref{sec:conclusions}.
\end{remark}

\section{Neural Network Parametrization}
\label{sec:nn_representation}

We use neural-network parametrizations for the active potential $\psi^{e}_{+}$ and the inverse damage resistance $R=r^{-1}$, both inferred indirectly from measured traction--separation data. The architectures are constructed to satisfy the required convexity and monotonicity conditions by construction. The physical normalizations, including a stress-free undeformed state for $\psi^{e}_{+}$ and zero damage at zero driving force for $R$, are also imposed exactly. Thus, the constitutive structure is embedded directly into the parametrization, without additional penalty terms in the data-misfit loss. The two network architectures are described below.

\subsection{Input-Convex Neural Network}
\label{sec:nn_representation:icnn}

The active potential is constructed from a raw ICNN potential
$\hat{\psi}$ \citep{amos2017input} with the active kinematic measures
$\boldsymbol{a}=(\pos{\delta_n},\lvert\delta_t\rvert)$ introduced in
\cref{sec:surf-free-energ} as inputs. The ICNN propagates the input through a
convex hidden path $\boldsymbol{z}$ together with direct input connections.
For $\ell=1,\dots,L$,
\begin{equation}
  \boldsymbol{z}^{(0)}
  =
  \sigma\!(
    \tilde{\boldsymbol{W}}_a^{(0)}\boldsymbol{a}
    +
    \boldsymbol{b}^{(0)}
  ),
  \qquad
  \boldsymbol{z}^{(\ell)}
  =
  \sigma\!(
    \tilde{\boldsymbol{W}}_z^{(\ell)}
    \boldsymbol{z}^{(\ell-1)}
    +
    \tilde{\boldsymbol{W}}_a^{(\ell)}
    \boldsymbol{a}
    +
    \boldsymbol{b}^{(\ell)}
  ),
  \qquad
  \hat{\psi}(\boldsymbol{a})
  =
  \boldsymbol{z}^{(L)},
  \label{eq:icnn_layer}
\end{equation}
where the output $\boldsymbol{z}^{(L)}$ is scalar. The constrained weights are
parameterized as
\begin{equation}
  \tilde{\boldsymbol{W}}_z^{(\ell)}
  =
  \mathrm{softplus}(\boldsymbol{W}_z^{(\ell)})
  \ge 0,
  \qquad
  \ell=1,\dots,L,
  \qquad
  \tilde{\boldsymbol{W}}_a^{(\ell)}
  =
  \mathrm{softplus}(\boldsymbol{W}_a^{(\ell)})
  \ge 0,
  \qquad
  \ell=0,\dots,L,
\end{equation}
and $\sigma$ is a convex and non-decreasing activation function, taken here as
softplus. These constraints make $\hat{\psi}$ jointly convex and
non-decreasing with respect to $\boldsymbol{a}$. In the original ICNN
construction \citep{amos2017input}, non-negativity is required for the
hidden-to-hidden weights, whereas the direct input weights may remain
unconstrained. We additionally constrain the direct input weights to be
non-negative so that the network is both convex and non-decreasing in
$\boldsymbol{a}$, as required by the present constitutive formulation.
The architecture of the ICNN in this work is illustrated in \cref{fig:icnn_scheme}.

\begin{figure}[htbp]
  \centering
  \includegraphics[width=0.5\linewidth]{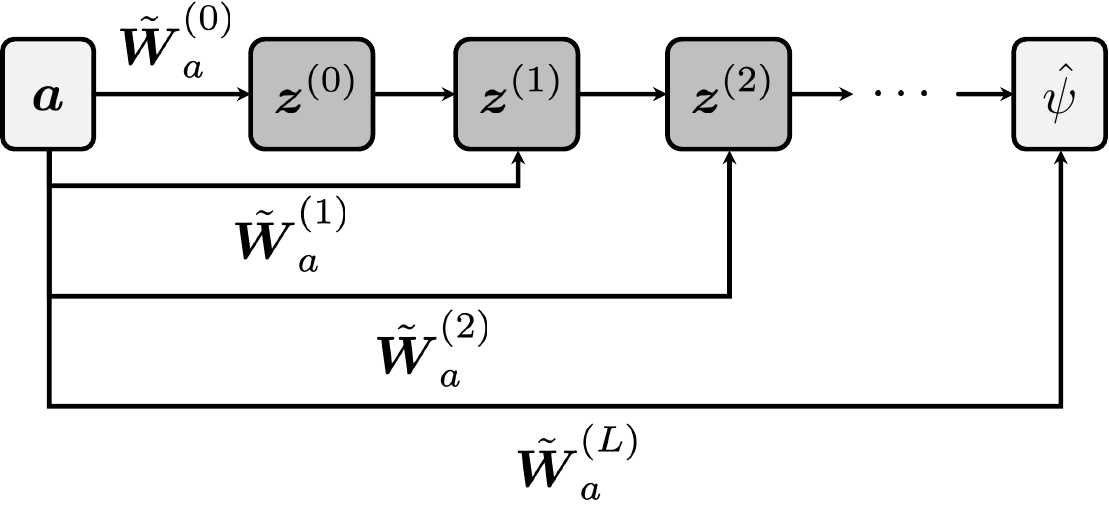}
  \caption{Schematic of the ICNN used to evaluate
  $\hat{\psi}(\boldsymbol{a})$.}
  \label{fig:icnn_scheme}
\end{figure}

A raw ICNN does not, in general, satisfy zero energy and zero traction at the
undeformed state. We impose these conditions through the first-order tangent
correction
\begin{equation}
  \psi^{e}_{+}(\boldsymbol{a})
  =
  \hat{\psi}(\boldsymbol{a})
  -
  \hat{\psi}(\boldsymbol{0})
  -
  \nabla_{\boldsymbol{a}}\hat{\psi}(\boldsymbol{0})^{\!\top}
  \boldsymbol{a}.
  \label{eq:tangent_subtraction}
\end{equation}
Since the correction is affine in $\boldsymbol{a}$, it preserves convexity; see \cref{lem:affine-composition}.
Moreover, for the positive-weight softplus ICNN used here,
$\nabla_{\boldsymbol{a}}\hat{\psi}$ is componentwise non-decreasing.
Therefore, for $\boldsymbol{a}\ge\boldsymbol{0}$,
\begin{equation}
  \nabla_{\boldsymbol{a}}\psi^{e}_{+}(\boldsymbol{a})
  =
  \nabla_{\boldsymbol{a}}\hat{\psi}(\boldsymbol{a})
  -
  \nabla_{\boldsymbol{a}}\hat{\psi}(\boldsymbol{0})
  \ge
  \boldsymbol{0},
\end{equation}
so the required non-decreasing dependence on the active kinematic measures is
also preserved. By construction,
\begin{equation}
  \psi^{e}_{+}(\boldsymbol{0})=0,
  \qquad
  \nabla_{\boldsymbol{a}}\psi^{e}_{+}(\boldsymbol{0})
  =
  \boldsymbol{0}.
\end{equation}
Finally, convexity of $\hat{\psi}$ implies that its tangent plane at
$\boldsymbol{a}=\boldsymbol{0}$ is a global underestimator. Hence,
$\psi^{e}_{+}(\boldsymbol{a})\ge0$. The resulting active potential is therefore
convex, non-negative, and non-decreasing, with zero energy and zero traction at
the undeformed state. Related normalization strategies have been used in
convex neural-network constitutive models to impose the stress-free reference
state \citep{thakolkaran2022nn,Flaschel2025}.

The learnable parameters are
\begin{equation}
  \boldsymbol{\theta}_{\psi}
  =
  \left\{
    \boldsymbol{W}_z^{(\ell)},
    \boldsymbol{W}_a^{(\ell)},
    \boldsymbol{b}^{(\ell)}
  \right\}_{\ell=1}^{L}
  \cup
  \left\{
    \boldsymbol{W}_a^{(0)},
    \boldsymbol{b}^{(0)}
  \right\}.
\end{equation}

The construction above permits coupling between the normal and tangential
responses. 
Many analytical cohesive laws \citep{camanho2003numerical,turon2006damage,alfano2001finite} instead take the elastic response to be uncoupled
so that tangential sliding does not alter the normal traction.
A generic ICNN does not enforce this structure and may therefore
learn artificial normal--tangential coupling from limited or noisy data. When
elastic mode coupling is not desired, we use a partitioned architecture in
which $\pos{\delta_n}$ and $\lvert\delta_t\rvert$ are passed through independent
ICNNs and their outputs are summed to form $\psi^{e}_{+}$. This construction
eliminates mixed normal--tangential derivatives by construction while
preserving the convexity and monotonicity properties described above.

\begin{remark}
\label[remark]{rem:icnn_smoothness}
The tangent correction preserves convexity for both the coupled and
partitioned architectures. A subtlety arises in smoothness with respect to
$\boldsymbol{\delta}$ because the inputs $\pos{\delta_n}$ and
$\lvert\delta_t\rvert$ are non-smooth at zero. At the pristine state, the
zero-gradient normalization removes the associated non-differentiability and
makes the composed potential differentiable. For a coupled ICNN, however, this
property does not generally extend along the coordinate axes. A discontinuity
in the normal traction may occur across $\delta_n=0$ when $\delta_t\neq0$, and
a discontinuity in the tangential traction may occur across $\delta_t=0$ when
$\delta_n>0$, because the corresponding derivatives of the active potential
need not vanish on these axes.
The partitioned form
$
\hat\psi
=
\hat\psi_n(\pos{\delta_n})
+
\hat\psi_t(\lvert\delta_t\rvert)
$
avoids these discontinuities. After the tangent correction, each contribution
has zero slope at the origin. Consequently, the composed potential
$\psi^{e}_{+}$ is globally convex and $C^1$ with respect to
$\boldsymbol{\delta}$.
\end{remark}

\subsection{Monotone Neural Network}
\label{sec:nn_representation:mnn}

We use an adaptation of the scalable monotonic neural network (SMNN) of
\citet{kim2024scalable} to parametrize the inverse resistance function
$R(Y,\beta)$. We refer to the resulting architectures as the MNN and
full-MNN.

For the MNN, $Y$ is treated as the monotone input, whereas $\beta$ remains
unrestricted. Following the partially connected structure of
\citet{kim2024scalable}, each hidden layer contains three streams: an
exponentiated stream, a confluence stream, and an unrestricted nonlinear
stream, denoted by the subscripts $e$, $c$, and $n$, respectively. The
first-layer inputs are
\begin{equation}
  \boldsymbol{x}^{(0)}_e=Y,
  \qquad
  \boldsymbol{x}^{(0)}_c
  =
  \boldsymbol{x}^{(0)}_n
  =
  \beta.
  \label{eq:mnn_input}
\end{equation}
For $\ell=1,\dots,L$, the streams are evaluated as
\begin{align}
  \boldsymbol{z}_e^{(\ell)}
  &=
  \tanh(
    \exp(\boldsymbol{W}_e^{(\ell)})
    \boldsymbol{x}_e^{(\ell-1)}
    +
    \boldsymbol{b}_e^{(\ell)}
  ),
  \\
  \boldsymbol{z}_c^{(\ell)}
  &=
  \tanh(
    \boldsymbol{W}_c^{(\ell)}
    \boldsymbol{x}_c^{(\ell-1)}
    +
    \boldsymbol{b}_c^{(\ell)}
  ),
  \\
  \boldsymbol{z}_n^{(\ell)}
  &=
  \tanh(
    \boldsymbol{W}_n^{(\ell)}
    \boldsymbol{x}_n^{(\ell-1)}
    +
    \boldsymbol{b}_n^{(\ell)}
  ).
  \label{eq:mnn_layer}
\end{align}
The streams are coupled between consecutive layers according to
\begin{equation}
  \boldsymbol{x}^{(\ell)}_e
  =
  \boldsymbol{z}^{(\ell)}_e
  \oplus
  \boldsymbol{z}^{(\ell)}_c,
  \qquad
  \boldsymbol{x}^{(\ell)}_c
  =
  \boldsymbol{x}^{(\ell)}_n
  =
  \boldsymbol{z}^{(\ell)}_n,
  \label{eq:mnn_coupling}
\end{equation}
where $\oplus$ denotes vector concatenation. The scalar output is
\begin{equation}
  g_{\boldsymbol{\theta}}(Y,\beta)
  =
  \exp(\boldsymbol{W}^{(L+1)})
  (
    \boldsymbol{z}_e^{(L)}
    \oplus
    \boldsymbol{z}_c^{(L)}
    \oplus
    \boldsymbol{z}_n^{(L)}
  )
  +
  b^{(L+1)}.
  \label{eq:mnn_out}
\end{equation}
The exponentiated weights are strictly positive. Together with the strictly
increasing $\tanh$ activation, this construction makes
$g_{\boldsymbol{\theta}}$ strictly increasing with respect to $Y$, while its
dependence on $\beta$ remains unrestricted. 
Here, $\tanh$ is adopted to obtain a smooth and strictly monotone parametrization.
The architecture of the MNN in this work is illustrated in \cref{fig:mnn_scheme}.

\begin{figure}[htbp]
  \centering
  \includegraphics[width=0.5\linewidth]{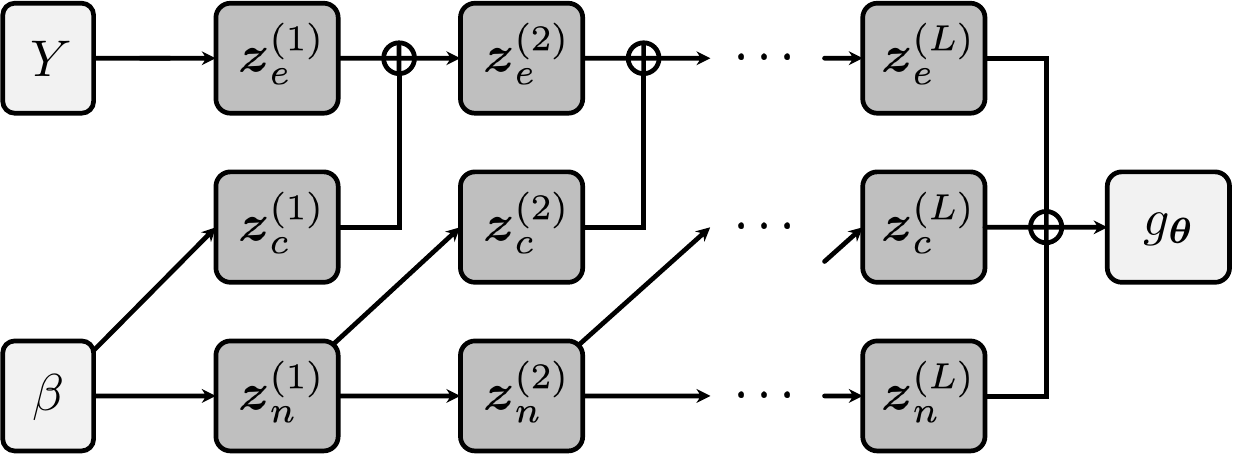}
  \caption{Schematic of the MNN used to evaluate
  $g_{\boldsymbol{\theta}}(Y,\beta)$.}
  \label{fig:mnn_scheme}
\end{figure}

The raw network output has an arbitrary offset and is not directly tied to the
damage scale. We therefore define the shifted output
\begin{align}
  &\hat R(Y,\beta)
  =
  g_{\boldsymbol{\theta}}(Y,\beta)
  -
  g_{\boldsymbol{\theta}}(0,\beta)
  -
  \operatorname{softplus}(b_0),
  \label{eq:mnn_inverse}
  \\
  &d_{\mathrm{tr}}
  =
  \operatorname{clamp}(\hat R(Y,\beta),0,1).
  \label{eq:mnn_trial_damage}
\end{align}
where $b_0$ is a learnable scalar. Since
\begin{equation}
  \hat R(0,\beta)
  =
  -\operatorname{softplus}(b_0)
  <
  0,
\end{equation}
the trial damage satisfies $d_{\mathrm{tr}}(0,\beta)=0$ for every mode
mixity. When the shifted network output crosses zero, damage initiates at the
mode-dependent threshold $Y_0(\beta)$ satisfying
\begin{equation}
  g_{\boldsymbol{\theta}}(Y_0,\beta)
  -
  g_{\boldsymbol{\theta}}(0,\beta)
  =
  \operatorname{softplus}(b_0).
\end{equation}
Thus, the offset introduces a strictly positive damage-onset threshold without
restricting the dependence of that threshold on $\beta$.

We also consider a more restrictive variant, referred to as the full-MNN, for
interfaces whose resistance increases with mode mixity. In this case, $R$ is required to increase with $Y$ and
decrease with $\beta$. Both variables can therefore be passed through the
monotone stream after replacing $\beta$ by $1-\beta$:
\begin{equation}
  \boldsymbol{x}^{(0)}_e
  =
  (Y,1-\beta).
\end{equation}
Since no unrestricted input remains, the confluence and nonlinear streams are
omitted. We define
\begin{equation}
  \hat R(Y,\beta)
  =
  g_{\boldsymbol{\theta}}(Y,1-\beta)
  -
  g_{\boldsymbol{\theta}}(0,1)
  -
  \operatorname{softplus}(b_0).
  \label{eq:fmnn_inverse}
\end{equation}

The reference value in \cref{eq:fmnn_inverse} is chosen to be independent of
$\beta$ so that the imposed monotonicity with respect to mode mixity is
preserved. Since $g_{\boldsymbol{\theta}}$ is increasing in $1-\beta$, this
choice also ensures $\hat R(0,\beta)<0$ and hence
$d_{\mathrm{tr}}(0,\beta)=0$ for all $\beta$. The resulting onset threshold
$Y_0(\beta)$ is non-decreasing with $\beta$ by construction.

The learnable parameters of the MNN are
\begin{equation}
  \boldsymbol{\theta}_{R}
  =
  \left\{
    \boldsymbol{W}_e^{(\ell)}, \boldsymbol{b}_e^{(\ell)},
    \boldsymbol{W}_c^{(\ell)}, \boldsymbol{b}_c^{(\ell)},
    \boldsymbol{W}_n^{(\ell)}, \boldsymbol{b}_n^{(\ell)}
  \right\}_{\ell=1}^{L}
  \cup
  \left\{
    \boldsymbol{W}^{(L+1)}, b^{(L+1)}, b_0
  \right\},
  \label{eq:mnn_params}
\end{equation}
whereas the full-MNN omits the confluence and nonlinear streams, so
\begin{equation}
  \boldsymbol{\theta}_{R}
  =
  \left\{
    \boldsymbol{W}_e^{(\ell)}, \boldsymbol{b}_e^{(\ell)}
  \right\}_{\ell=1}^{L}
  \cup
  \left\{
    \boldsymbol{W}^{(L+1)}, b^{(L+1)}, b_0
  \right\}.
  \label{eq:fmnn_params}
\end{equation}

\begin{remark}
\label[remark]{rem:univariate_mnn}
When mode mixity is absent, such as in a one-dimensional cohesive problem, the
unrestricted streams are unnecessary. The architecture then reduces to a
univariate monotone neural network with positive weights and non-decreasing
activation functions \citep{lang2005monotonic,daniels2010monotone}.
\end{remark}

\section{Dataset and Identification Objective}
\label{sec:training}

Because damage evolves irreversibly, the separation--traction measurements
cannot be treated as independent samples. The calibration dataset is therefore
organized into $N_{\mathrm{path}}$ ordered loading--unloading paths. Each path
$p=1,\ldots,N_{\mathrm{path}}$ is represented as
\begin{equation}
  \mathcal{P}_p
  =
  \left\{
    \left(
      \boldsymbol{\delta}^{(m)}_p,
      \boldsymbol{t}^{(m)}_p
    \right)
  \right\}_{m=0}^{N_p},
  \label{eq:loading_path_data}
\end{equation}
and the complete dataset is
\begin{equation}
  \mathcal{S}
  =
  \left\{
    \mathcal{P}_p
  \right\}_{p=1}^{N_{\mathrm{path}}}.
\end{equation}
Each path is assumed to start from a pristine interface,
$d_p^{(0)}=0$. The damage states are not treated as independent calibration
variables; they are reconstructed sequentially from the separation history
using the constitutive update. Consequently, the predicted traction
$\widehat{\boldsymbol{t}}_p^{(m)}$ depends on the preceding loading history.

Let
$\boldsymbol{\theta}
\coloneqq
\left(
  \boldsymbol{\theta}_{\psi},
  \boldsymbol{\theta}_{R}
\right)$
collect the parameters of the active potential and inverse resistance. They
are identified by minimizing
\begin{equation}
  \mathcal{L}(\boldsymbol{\theta})
  =
  \sum_{p=1}^{N_{\text{path}}}
    \frac{1}{N_p}
    \sum_{m=1}^{N_p}
    \frac{
      |
        \widehat{\boldsymbol{t}}_p^{(m)}
        -
        \boldsymbol{t}_p^{(m)}
      |_2^2
    }{
      (t_p^{\max})^2
    },
  \label{eq:loss}
\end{equation}
where
$t_p^{\max}=\max_m\|\boldsymbol{t}_p^{(m)}\|_{\infty}$ and the expectation is
taken uniformly over the loading paths. In the present implementation, each
optimization step evaluates all paths over their complete ordered histories
before backpropagation. The loss is minimized using the Adam optimizer.

\begin{remark}
The decomposition of the constitutive response into the active potential
$\psi^{e}_{+}$ and the damage resistance $r$ is generally non-unique:
distinct $(\psi^{e}_{+},r)$ pairs may reproduce the same traction histories,
as illustrated in \cref{app:mixed-mode-result}. Consequently,
traction--separation data do not, in general, uniquely identify the two latent
constitutive functions individually.
\end{remark}

\section{Numerical Examples}
\label{sec:examples}

This section presents numerical examples that assess the representational and
identification properties of the proposed cohesive formulation. The first
example considers Mode~I loading and examines whether the formulation can
represent cohesive laws with sharp changes in the traction--separation
response. It further tests the identified model under several successive
unload--reload cycles with nonlinear unloading behavior to assess irreversible
damage evolution and non-negative dissipation. The second example considers
mixed-mode identification when calibration data are available only for the two
pure modes and investigates how the predicted intermediate response depends on
the constitutive structure imposed on the damage resistance. It also examines
the sensitivity of the identification and resulting predictions to network
size and random initialization. The third example considers several classical
mixed-mode cohesive laws with distinct analytical forms and evaluates
prediction at held-out mode mixities. The identified model is further tested
under non-proportional loading to examine its response when the mode mixity
evolves along the loading path.

All neural-network models and training procedures are implemented in PyTorch
\citep{paszke2019pytorch}. The network architectures, optimization
hyperparameters, and implementation details are summarized in
\cref{app:hyperparams}.

\subsection{Pure Mode~I Identification}
\label{sec:ex-mode-I}

This example considers several synthetic Mode~I traction--separation responses
with different loading and unloading behaviors to assess the flexibility and
thermodynamic consistency of the proposed formulation. Shear deformation is
not considered, so the problem reduces to a one-dimensional cohesive response.

We first synthesize data from two classical Mode~I cohesive laws --- a trapezoidal law with a constant-strength plateau and an exponential law with a long softening tail --- both exhibiting a linear increase in traction up to the peak. The plateau and extended softening tail make these responses challenging to represent accurately using fixed-form classical models, including versatile families such as PPR.
The trained model, as shown in \cref{fig:flex}, reproduces the target responses closely. The sharp corners and the smooth tails are automatically captured by the learned inverse resistance combined with the ICNN elastic potential.

\renewcommand{\tempvar}{0.35}
\begin{figure}[htbp]
  \centering
  \begin{subfigure}[b]{\tempvar\linewidth}\centering
    \includegraphics[width=\linewidth]{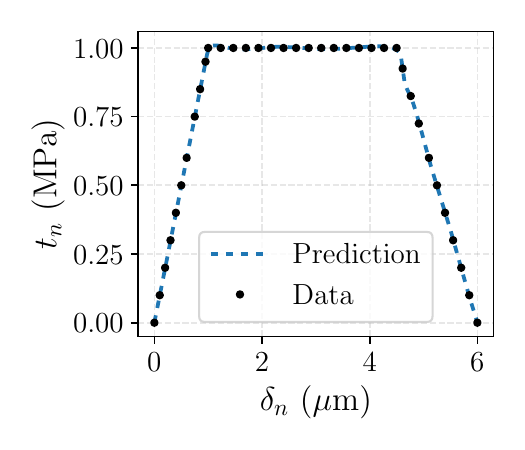}
    \caption{Trapezoidal}\end{subfigure}%
  \hspace{0.03\linewidth}%
  \begin{subfigure}[b]{\tempvar\linewidth}\centering
    \includegraphics[width=\linewidth]{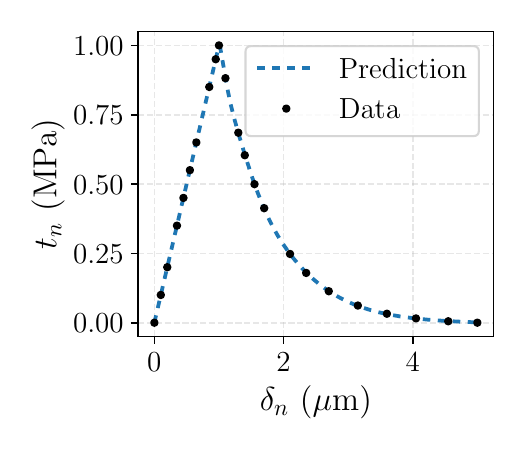}
    \caption{Exponential}\end{subfigure}
    \caption{Mode~I traction--separation responses for (a) trapezoidal and (b) exponential-tail cohesive laws.}
  \label{fig:flex}
\end{figure}

We further consider a nonlinear loading envelope with a concave-to-convex
\textit{S}-shaped unloading response. The model is trained on the monotonic
loading envelope, two pre-peak unload--reload cycles that retrace the loading
response, and one post-peak unload--reload cycle. The training and evaluation
separation trajectories
are shown in \cref{fig:dissipation_multiunload}(a), with the
corresponding traction--separation responses in panel~(b). The evaluation
history contains five unload--reload cycles with progressively increasing
maximum separation, each unloaded to the origin.

The first two evaluation cycles originate from the pre-peak branch and retrace
the elastic loading response. The remaining three originate from the post-peak
softening branch and exhibit nonlinear \textit{S}-shaped unloading responses.
Since only one post-peak cycle is included in the training data, the other
post-peak branches evaluate the response predicted by the identified model at
different damage states.

The cumulative mechanical dissipation is obtained from the intrinsic
dissipation in \cref{eq:dissipation} as
\begin{equation}
  D(t)
  =
  \int_0^t
  Y(\tau)\,\dot d(\tau)\,\dd\tau .
  \label{eq:cum_dissipation}
\end{equation}
As shown in \cref{fig:dissipation_D}, both $d$ and $D$ remain constant during
unloading and reloading below the previously attained maximum separation.
They increase only when the loading path reaches a new maximum separation.
Thus, the learned response satisfies damage irreversibility and non-negative
mechanical dissipation throughout the loading history.

\begin{figure}[htbp]
  \centering
  \begin{subfigure}[b]{0.30\linewidth}
    \centering
    \includegraphics[width=\linewidth]{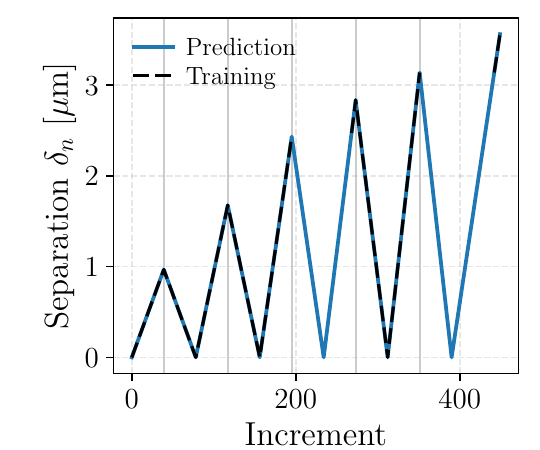}
    \caption{}
    \label{fig:dissipation_loading}
  \end{subfigure}
  \hfill
  \begin{subfigure}[b]{0.30\linewidth}
    \centering
    \includegraphics[width=\linewidth]{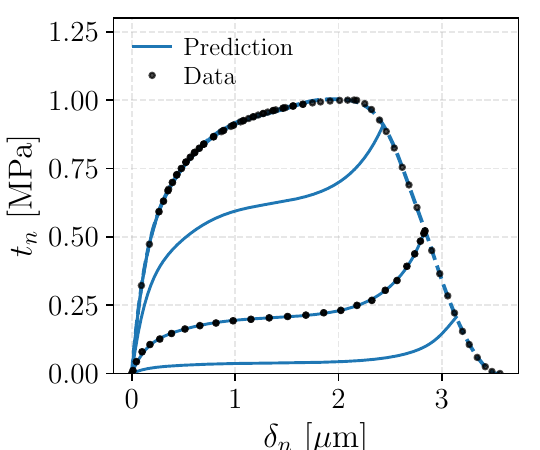}
    \caption{}
    \label{fig:multiunload_response}
  \end{subfigure}
  \hfill
  \begin{subfigure}[b]{0.3542\linewidth} 
    \centering
    \includegraphics[width=\linewidth]{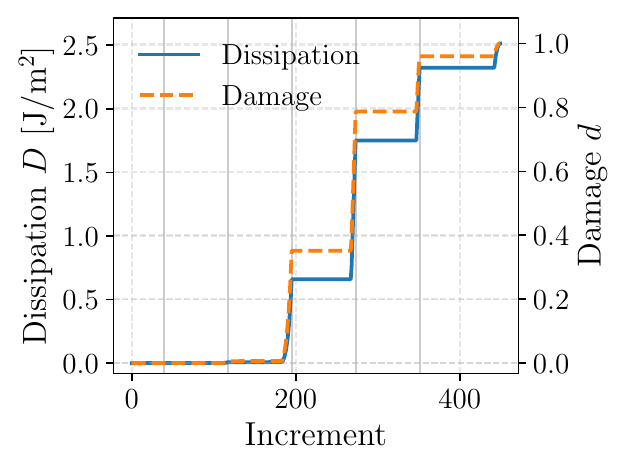}
    \caption{}
    \label{fig:dissipation_D}
  \end{subfigure}

    \caption{Mode~I response under five unload--reload cycles followed by complete separation. (a) Prescribed evaluation history (solid blue) together with the loading segments used for training (dashed black). (b) Predicted traction--separation response (solid blue) compared with the training data (black dots). (c) Cumulative mechanical dissipation $D$ (solid blue, left axis) and damage $d$ (dashed orange, right axis) predicted by the learned model.}
  \label{fig:dissipation_multiunload}
\end{figure}

\subsection{Mixed-Mode Identification from Pure-Mode Data}
\label{sec:examp-MNN-MRQS}

This example examines whether the proposed formulation can identify a mixed-mode cohesive response when calibration data are available only for the two pure modes. We compare the MNN and full-MNN parametrizations of the inverse resistance in terms of pure-mode accuracy, sensitivity to network capacity and initialization, and interpolation across unobserved mode mixities. Both models are trained on a synthetic irregular Mode~I response and a half-sine Mode~II response, with no intermediate mixed-mode data included during calibration. The two pure-mode responses are intentionally chosen to be markedly different, providing a nontrivial test of mixed-mode interpolation.

For each inverse-resistance architecture, we consider three network sizes and
train each configuration from five random initializations. In every case, the
active-potential ICNN and the inverse-resistance network are trained jointly
using the combined Mode~I and Mode~II loading data. The ICNN architecture is
kept identical across all configurations so that the comparison isolates the
effect of the MNN and full-MNN parametrizations of the inverse resistance. For corresponding hidden-layer widths, the full-MNN contains fewer parameters
because the unrestricted streams are omitted. The resulting pure-mode
predictions are compared below.

The resulting Mode~I predictions are shown in
\cref{fig:mnn_mrqs_mm_modeI_ts}. At the smaller network capacities, the
predictions exhibit noticeable sensitivity to initialization, and some
realizations do not reproduce the irregular features of the target response.
Increasing the resistance-network capacity reduces the variability across
initializations and improves agreement with the data. At the largest
configuration, the realizations for both the MNN and full-MNN nearly collapse
onto the target Mode~I response.

\begin{figure}[htbp]
  \centering
  \begin{subfigure}[b]{\textwidth}
    \centering
    \includegraphics[width=0.32\linewidth]{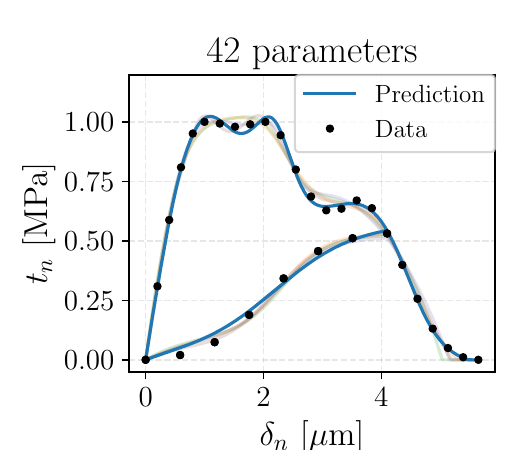}\hfill
    \includegraphics[width=0.32\linewidth]{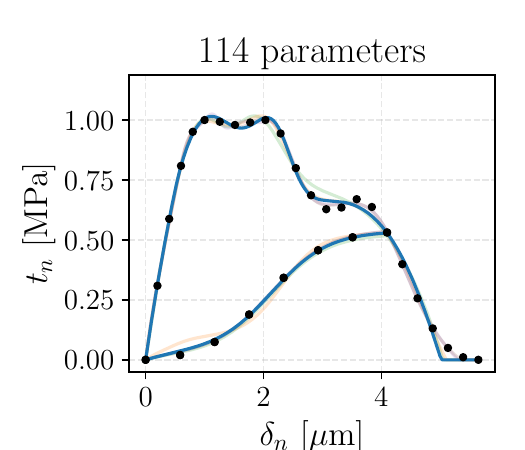}\hfill
    \includegraphics[width=0.32\linewidth]{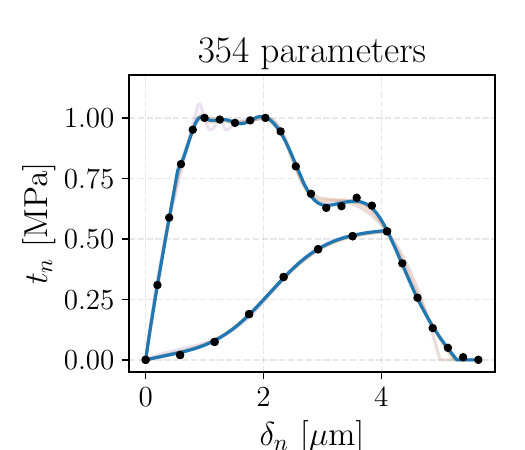}
    \caption{MNN}
    \label{fig:mnn_mrqs_mm_modeI_mnn}
  \end{subfigure}\\[1ex]
  \begin{subfigure}[b]{\textwidth}
    \centering
    \includegraphics[width=0.32\linewidth]{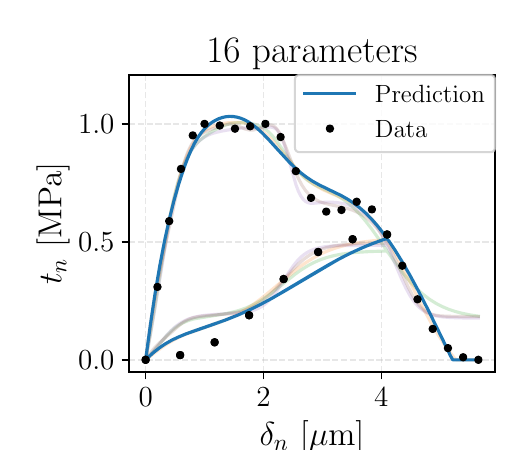}\hfill
    \includegraphics[width=0.32\linewidth]{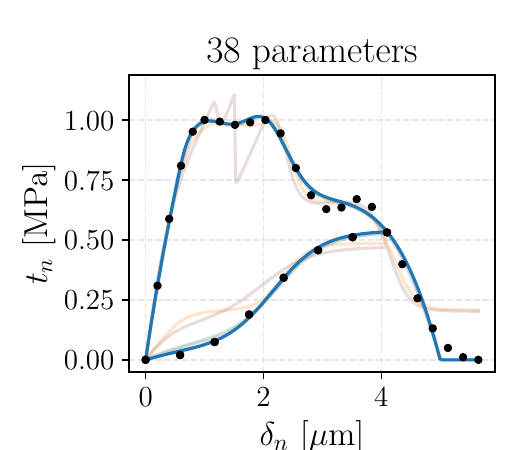}\hfill
    \includegraphics[width=0.32\linewidth]{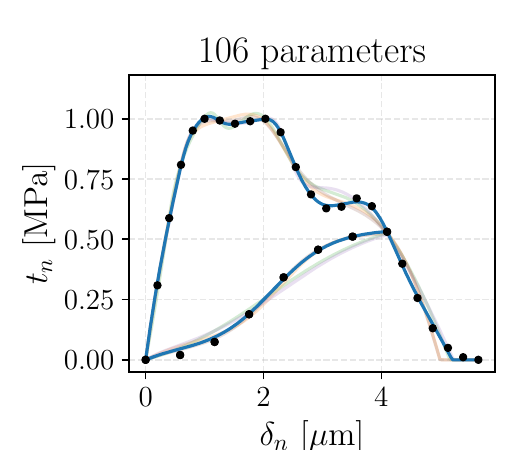}
    \caption{full-MNN}
    \label{fig:mnn_mrqs_mm_modeI_fmnn}
  \end{subfigure}
  \caption{Mode~I traction--separation predictions obtained with (a) MNN and
  (b) full-MNN inverse-resistance parametrizations. Three resistance-network
  sizes are considered, with five random initializations for each size. Panel
  titles give the number of resistance-network parameters. The solid blue
  curve denotes the lowest-loss realization, the faint curves denote the
  remaining initializations, and the black dots denote the training data.}
  \label{fig:mnn_mrqs_mm_modeI_ts}
\end{figure}

The corresponding Mode~II results are shown in
\cref{fig:mnn_mrqs_mm_modeII_ts}. The smaller configurations again exhibit
substantial sensitivity to initialization, with several realizations producing
poor or irregular responses. This variability decreases markedly as the
network capacity increases. At the largest configuration, both architectures
closely reproduce the Mode~II data with little dependence on initialization.
For comparable hidden-layer widths, the full-MNN uses substantially fewer
resistance-network parameters than the MNN; at the largest configuration, it
uses 106 parameters compared with 354 for the MNN while attaining comparable
pure-mode accuracy and reproducibility. 

\begin{figure}[htbp]
  \centering
  \begin{subfigure}[b]{\textwidth}
    \centering
    \includegraphics[width=0.32\linewidth]{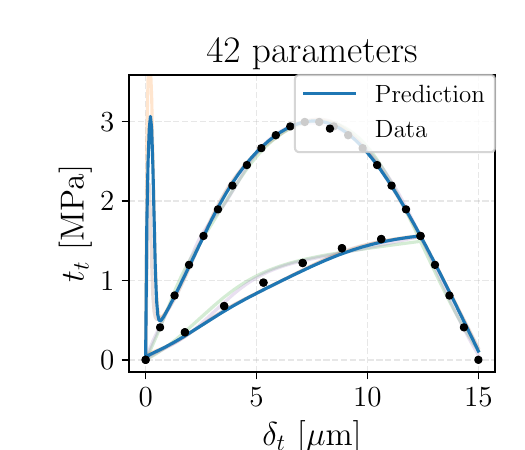}\hfill
    \includegraphics[width=0.32\linewidth]{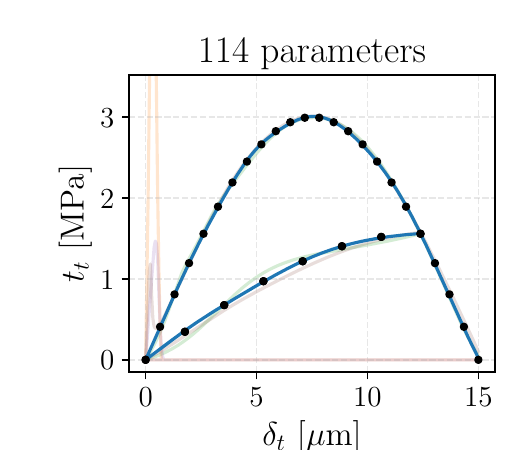}\hfill
    \includegraphics[width=0.32\linewidth]{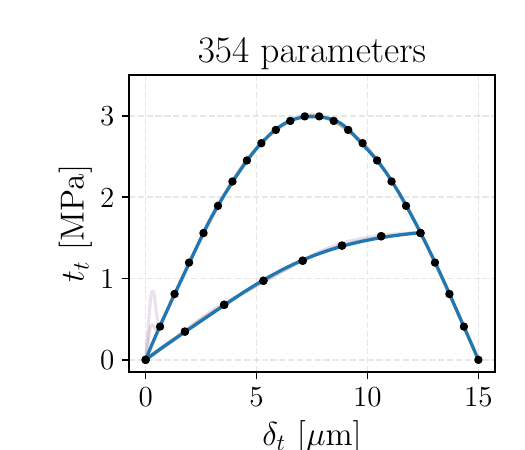}
    \caption{MNN}
    \label{fig:mnn_mrqs_mm_modeII_mnn}
  \end{subfigure}\\[1ex]
  \begin{subfigure}[b]{\textwidth}
    \centering
    \includegraphics[width=0.32\linewidth]{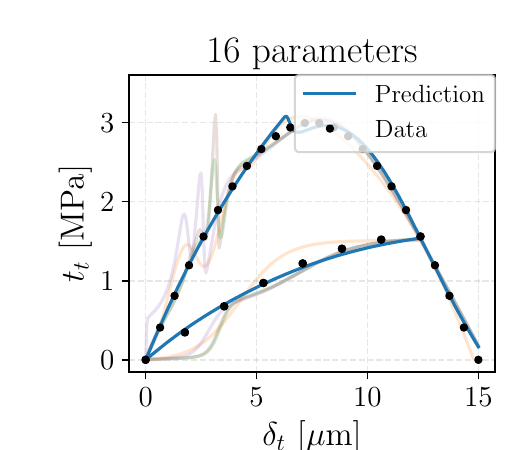}\hfill
    \includegraphics[width=0.32\linewidth]{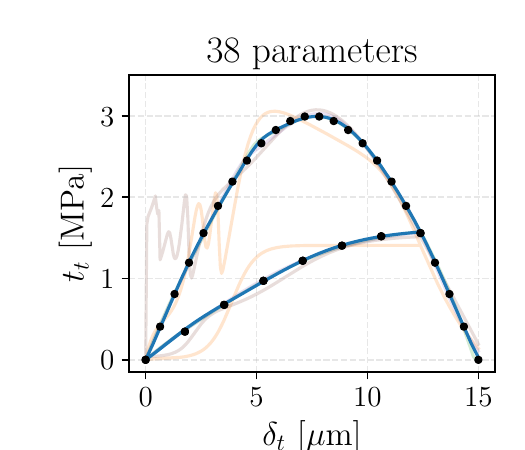}\hfill
    \includegraphics[width=0.32\linewidth]{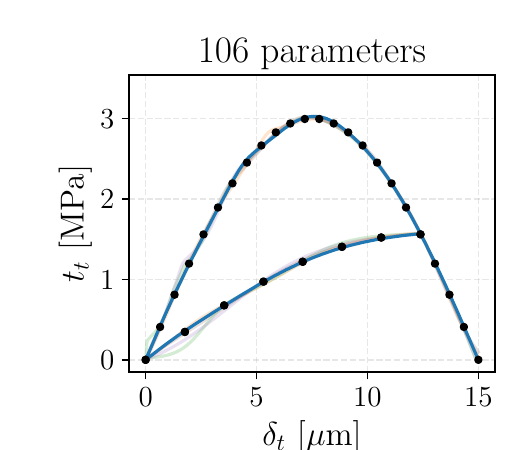}
    \caption{full-MNN}
    \label{fig:mnn_mrqs_mm_modeII_fmnn}
  \end{subfigure}
  \caption{Mode~II traction--separation predictions obtained with (a) MNN and
  (b) full-MNN inverse-resistance parametrizations. Three resistance-network
  sizes are considered, with five random initializations for each size. Panel
  titles give the number of resistance-network parameters. The solid blue
  curve denotes the lowest-loss realization, the faint curves denote the
  remaining initializations, and the black dots denote the training data.}
  \label{fig:mnn_mrqs_mm_modeII_ts}
\end{figure}

The more consequential difference between the two parametrizations appears at
intermediate mode mixities, which are not represented in the calibration data.
We evaluate the highest-capacity models under proportional mixed-mode loading
for eight values of $\beta\in[0,1]$. The traction and separation magnitudes are
defined as $\lVert\boldsymbol{t}\rVert=\sqrt{t_n^2+t_t^2}$ and
$\lVert\boldsymbol{\delta}\rVert=\sqrt{\delta_n^2+\delta_t^2}$, respectively.
The predicted traction--separation responses and corresponding fracture
energies are shown in \cref{fig:mixed_mode_beta_sweep_ts}.

Both models recover the prescribed pure-mode responses at $\beta=0$ and
$\beta=1$, but their predictions between these limits differ substantially.
Because the dependence of the MNN on $\beta$ is unrestricted, the intermediate
responses are weakly constrained by the pure-mode data alone. Consequently,
the predicted traction--separation curves may vary irregularly and
non-monotonically with mode mixity, as also reflected in the non-monotone
$G_c(\beta)$. These responses remain thermodynamically admissible under the
constitutive assumptions of the model, but thermodynamic admissibility alone
does not ensure that an interpolation is physically representative where no
data are available. The result therefore illustrates both the flexibility of
the MNN and the limited identifiability of mixed-mode behavior from pure-mode
data alone.

The full-MNN introduces the additional assumption that the resistance
increases with $\beta$. This inductive bias constrains the otherwise
underdetermined interpolation and guarantees a non-decreasing fracture energy,
but it represents prior constitutive information rather than behavior inferred
from intermediate-mixity data. Reliable identification of the mixed-mode
response therefore requires either suitable prior structure or calibration
data spanning the relevant range of mode mixities. The latter setting is
examined in the next example, where data at several intermediate mixities are
included in the calibration and predictions are evaluated at held-out
mixities. Additional results for the present example, including the learned
constitutive functions and training histories, are provided in
\cref{app:mixed-mode-result}.

\begin{figure}[htbp]
  \centering
  \begin{subfigure}[b]{\textwidth}
    \centering
    \includegraphics[width=0.366\linewidth]{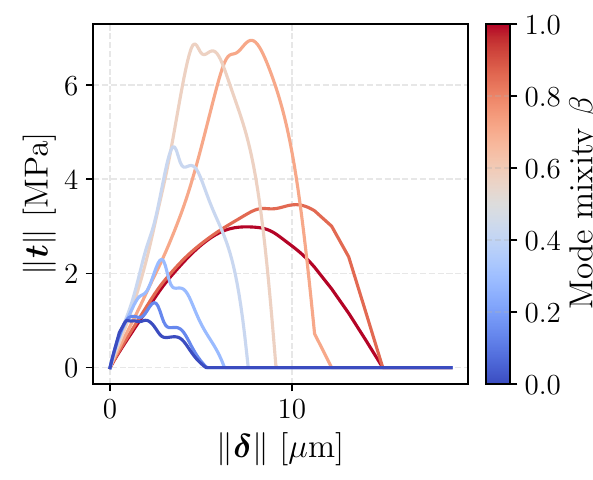}
    \hspace{0.04\linewidth}
    \includegraphics[width=0.30\linewidth]{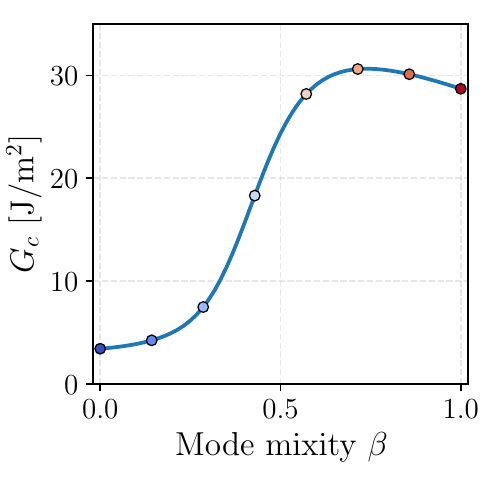}
    \caption{MNN}
    \label{fig:beta_sweep_mnn}
  \end{subfigure}\\[1ex]
  \begin{subfigure}[b]{\textwidth}
    \centering
    \includegraphics[width=0.366\linewidth]{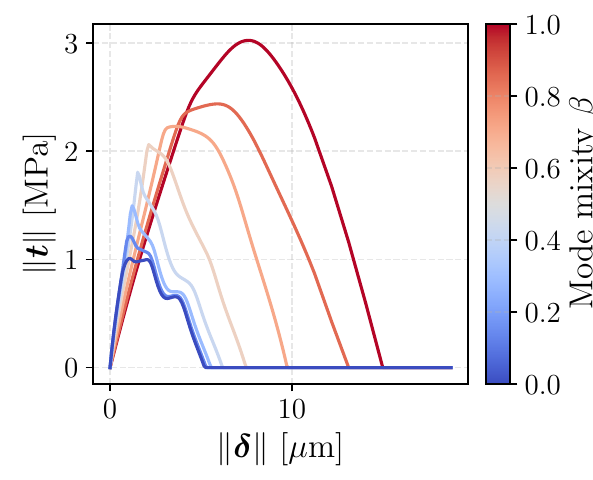}
    \hspace{0.04\linewidth}
    \includegraphics[width=0.30\linewidth]{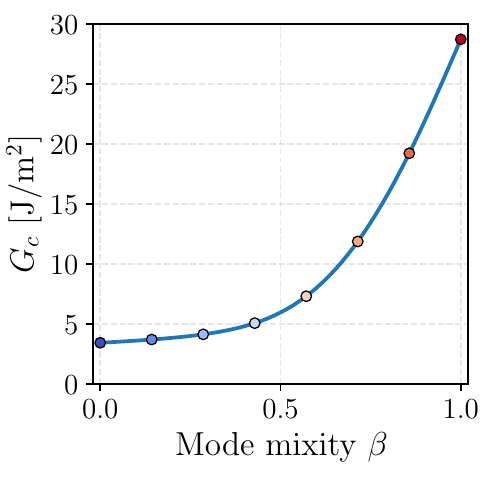}
    \caption{full-MNN}
    \label{fig:beta_sweep_fmnn}
  \end{subfigure}
  \caption{Mixed-mode responses predicted by (a) the MNN and (b) the full-MNN
  after calibration using only pure Mode~I and Mode~II data. The left column
  shows traction magnitude $\lVert\boldsymbol{t}\rVert$ versus separation
  magnitude $\lVert\boldsymbol{\delta}\rVert$ for eight mode mixities
  $\beta\in[0,1]$. The right column shows the corresponding fracture energy
  $G_c(\beta)$. }
  \label{fig:mixed_mode_beta_sweep_ts}
\end{figure}

\subsection{Mixed-Mode Identification Across Multiple Mode Mixities}
\label{sec:examp-dimitri}

This example evaluates the proposed formulation for mixed-mode cohesive laws
using calibration data from multiple mode mixities. The models are calibrated
at selected mode mixities and evaluated at held-out intermediate mixities to
assess interpolation across the mode-mixity range. We further examine one
identified model under non-proportional loading, for which the mode mixity
evolves along the loading path, to assess the resulting traction and
dissipation histories beyond the proportional trajectories used for
calibration.

Data are synthesized from three mixed-mode cohesive laws: an exponential law
(CZM-Exp) \citep{van2006improved}, a bilinear law with a power-law mixed-mode
criterion (CZM-BL) \citep{camanho2003numerical}, and the polynomial potential
PPR law \citep{Park2009}. The corresponding traction--separation relations and
parameter values are summarized in \cref{app:dimitri_laws}. These laws provide three distinct response classes: smooth exponential softening, a bilinear law
with a sharp strength transition, and a mode-dependent polynomial response. Because these models are based on distinct analytical cohesive constructions, they provide a direct test of whether the proposed framework can represent different cohesive-law families within the same constitutive structure.

For each law, data are generated along thirteen evenly spaced mode mixities
$\beta$. Ten mode mixities are used for calibration and three are withheld for
testing. Each configuration is trained from four random initializations for
CZM-Exp and CZM-BL and from eight initializations for PPR; representative
results are reported. 
The inverse resistance is represented by the MNN.

\Cref{fig:dimitri_czm-exp} show the normal
traction, tangential traction, and fracture energy for the three cohesive laws.
The identified models closely reproduce the calibration responses and
accurately predict the held-out mode mixities, as quantified in
\cref{tab:dimitri_mse}.

\begin{figure}[htbp]
  \centering
  \begin{subfigure}[b]{\textwidth}\centering
    \includegraphics[width=0.31\linewidth]{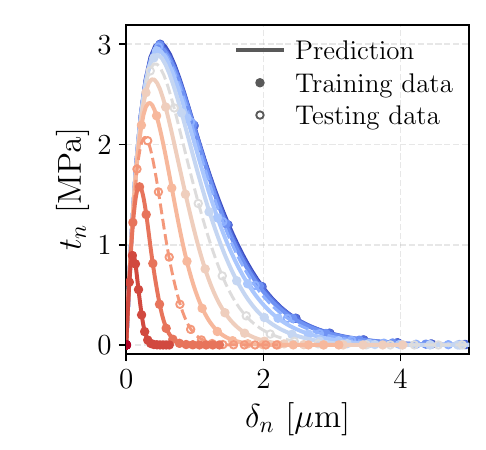}\hfill
    \includegraphics[width=0.37\linewidth]{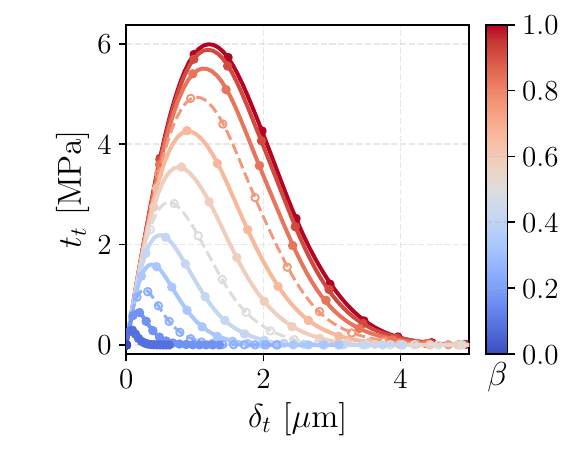}\hfill
    \includegraphics[width=0.31\linewidth]{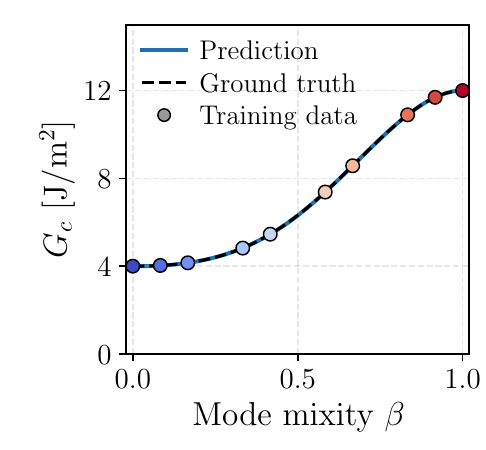}
    \caption{CZM-Exp Data}\label{fig:dimitri_czm-exp_mnn}
  \end{subfigure}\\[1ex]
  \begin{subfigure}[b]{\textwidth}\centering
    \includegraphics[width=0.31\linewidth]{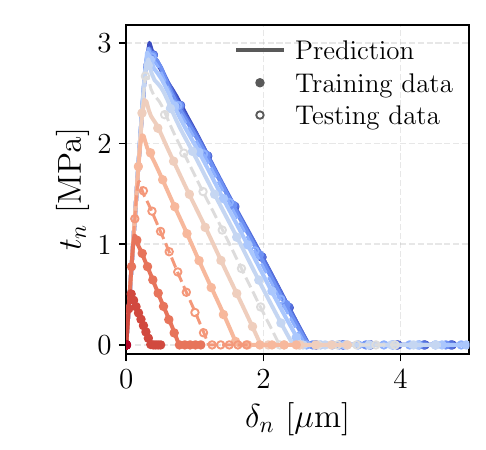}\hfill
    \includegraphics[width=0.37\linewidth]{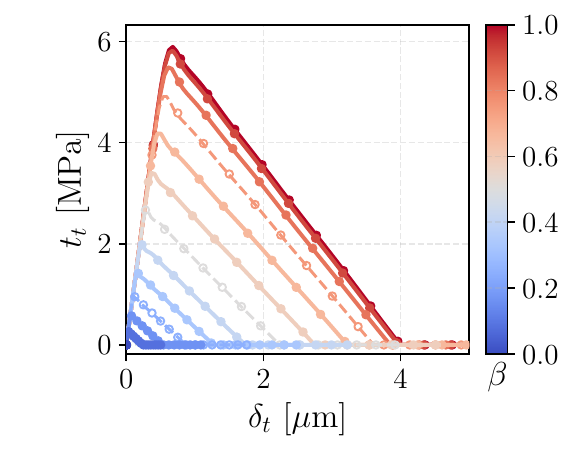}\hfill
    \includegraphics[width=0.31\linewidth]{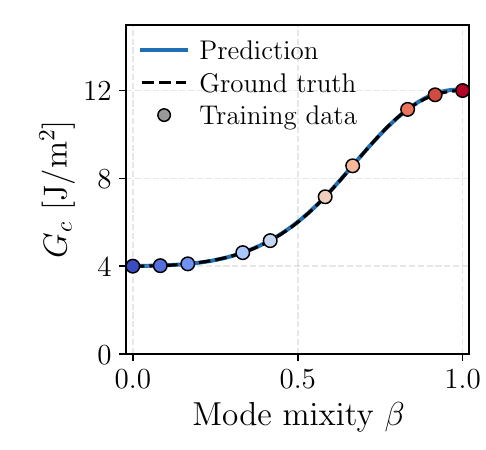}
    \caption{CZM-BL Data}\label{fig:dimitri_czm-bl_mnn}
  \end{subfigure} \\[1ex]
  \begin{subfigure}[b]{\textwidth}\centering
    \includegraphics[width=0.31\linewidth]{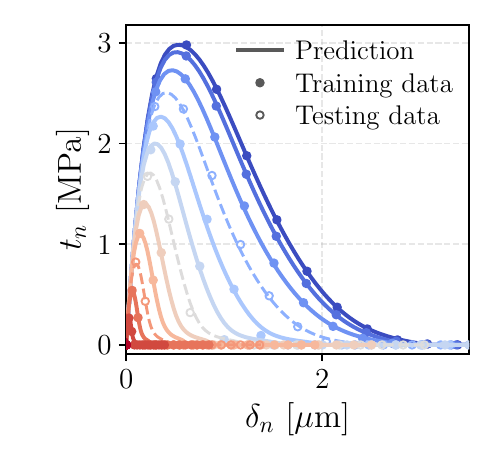}\hfill
    \includegraphics[width=0.37\linewidth]{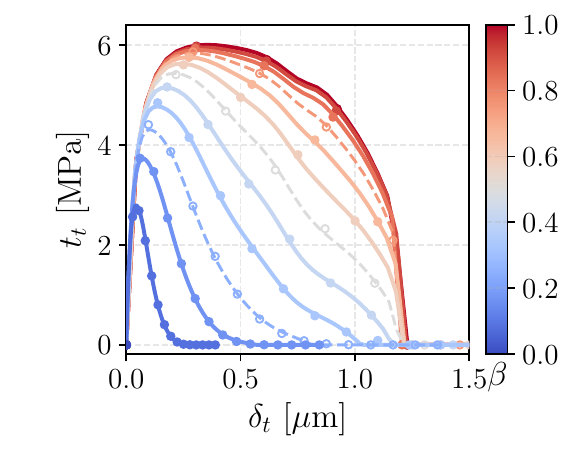}\hfill
    \includegraphics[width=0.31\linewidth]{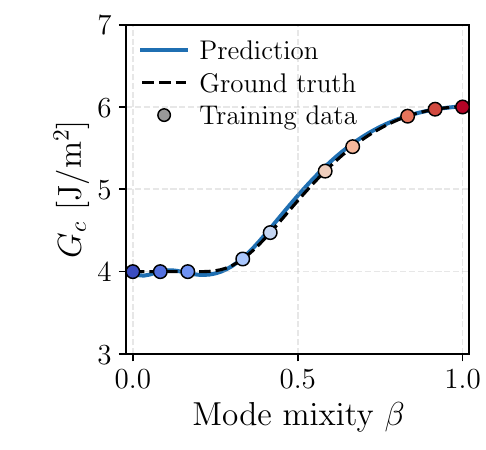}
    \caption{ PPR Data}\label{fig:dimitri_ppr_mnn}
  \end{subfigure}
    \caption{Mixed-mode identification results for (a) CZM-Exp, (b) CZM-BL, and
    (c) PPR. The left and center columns show the normal and tangential
    traction--separation responses, respectively, with color indicating the mode
    mixity $\beta$. Filled markers denote calibration data and open markers denote
    the held-out mixities $\beta\in\{0.25,0.5,0.75\}$; the corresponding model
    predictions are shown by the curves. The right column compares the predicted
    fracture energy $G_c(\beta)$ with the analytical reference.}
  \label{fig:dimitri_czm-exp}
\end{figure}

\begin{table}[htbp]
\centering
\small
\caption{Training and held-out test errors for the three models shown in \cref{fig:dimitri_czm-exp}. Each entry is the normalized MSE averaged over the $10$ training mixities or the $3$ held-out test mixities $\beta\in\{0.25,0.5,0.75\}$.}
\label{tab:dimitri_mse}
\begin{tabular}{@{}lccc@{}}
\toprule
 & CZM-Exp & CZM-BL & PPR \\
\midrule
Train & $2.9\times10^{-6}$ & $7.4\times10^{-6}$ & $3.2\times10^{-5}$ \\
Test  & $4.7\times10^{-6}$ & $8.6\times10^{-6}$ & $5.2\times10^{-5}$ \\
\bottomrule
\end{tabular}
\end{table}

The trajectories considered above are proportional, so the mode mixity $\beta$
remains constant along each loading path. We further examine the model
identified from CZM-Exp under non-proportional loading, although no
non-proportional trajectories are included in the calibration data.
\Cref{fig:dimitri_nonprop} considers two sequential paths in the
$\delta_n$--$\delta_t$ separation plane. The N$\to$T path first opens the
interface in Mode~I to $\delta_n=2\,\mu\mathrm{m}$ and then increases
$\delta_t$ at fixed $\delta_n$. The T$\to$N path applies the same two loading
stages in the opposite order. Both paths terminate at
$(\delta_n,\delta_t)=(2,2)\,\mu\mathrm{m}$. The responses are plotted against
the arc length $s=\int\lVert\dd\boldsymbol{\delta}\rVert$ traversed along each
path.

The traction histories depend strongly on the loading sequence. Along
N$\to$T, the first stage corresponds to Mode~I loading, with $t_t=0$ while
$t_n$ rises and subsequently decreases. During the second stage, the
tangential traction develops as $\delta_t$ increases, while the normal traction
continues to decrease despite the fixed $\delta_n$. The T$\to$N path exhibits
the corresponding interaction when normal opening is applied after the initial
shearing stage. These changes reflect the combined effects of mixed-mode
coupling and damage evolution. In particular, because a single scalar damage
variable degrades the entire active potential, damage generated by either
loading component contributes to degradation of both active traction
components.

Despite their different histories, the two paths reach the same final damage
in the present test. Because they also terminate at the same separation, the
constitutive state $(\boldsymbol{\delta},d)$ is identical at the endpoint, and
the predicted terminal traction vectors coincide. The cumulative dissipation,
however, differs substantially between the two paths, as shown in
\cref{fig:dimitri_nonprop_D}. From \cref{eq:cum_dissipation},
$D=\int Y\,\dd d$, so the dissipated energy depends on the values of the
driving force at which the successive damage increments occur. Since
$Y=r(d,\beta)$ during active damage growth, changing the order of normal and
tangential loading changes the sequence of mode mixities and driving forces
encountered during damage evolution. Consequently, the same final
$(\boldsymbol{\delta},d)$ can be reached with different cumulative
dissipation. Since both paths attain the same final damage, the larger value of
$D$ for T$\to$N indicates that its damage increments occur at a larger
damage-weighted average driving force.

\begin{remark}
\label[remark]{rem:path-dependence}
Within the present scalar-damage formulation, a mode-dependent resistance
generally leads to path-dependent cumulative dissipation under
non-proportional loading. Suppose instead that the cumulative dissipation were
a state function, $D=\hat D(\boldsymbol{\delta},d)$. Variations of
$\boldsymbol{\delta}$ at fixed $d$ are non-dissipative, so
$\partial\hat D/\partial\boldsymbol{\delta}=0$ along such admissible
variations and $\hat D$ can depend only on $d$. During active damage growth,
\cref{eq:cum_dissipation} and the consistency condition give
\begin{equation*}
  \dd D
  =
  Y\,\dd d
  =
  r(d,\beta)\,\dd d,
\end{equation*}
and therefore
\begin{equation*}
  \hat D'(d)=r(d,\beta).
\end{equation*}
If the same damage state can be attained by proportional loading at different
mode mixities, the left-hand side is independent of $\beta$, so the resistance
must also be independent of $\beta$. It would then follow that
\begin{equation*}
  G_c(\beta)
  =
  \int_0^1 r(d)\,\dd d
\end{equation*}
is independent of mode mixity. Thus, within the constitutive structure adopted
here, a mode-dependent fracture energy represented through $r(d,\beta)$
generally entails path-dependent dissipation when $\beta$ evolves during
damage growth.
\end{remark}

\renewcommand{\tempvar}{0.25}
\begin{figure}[htbp]
  \centering
  \begin{subfigure}[b]{\tempvar\linewidth}\centering
    \includegraphics[width=\linewidth]{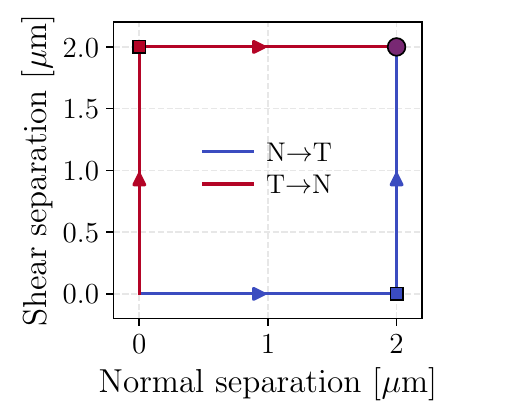}
    \caption{}\label{fig:dimitri_nonprop_paths}
  \end{subfigure}\hspace{0.05\linewidth}%
  \begin{subfigure}[b]{\tempvar\linewidth}\centering
    \includegraphics[width=\linewidth]{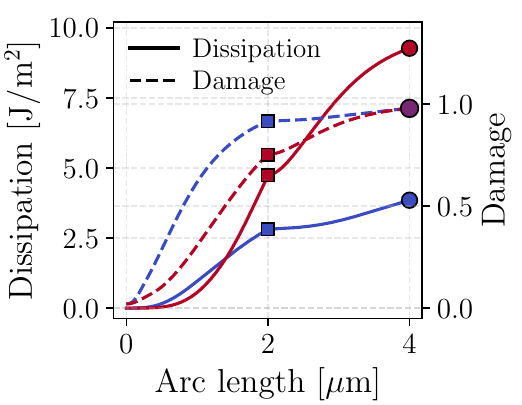}
    \caption{}\label{fig:dimitri_nonprop_D}
  \end{subfigure}\\[1ex]
  \begin{subfigure}[b]{\tempvar\linewidth}\centering
    \includegraphics[width=\linewidth]{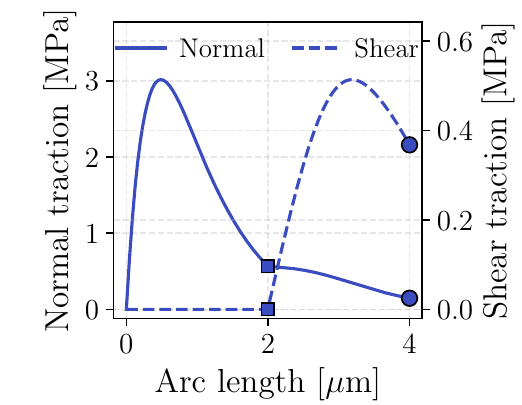}
    \caption{}\label{fig:dimitri_nonprop_tNT}
  \end{subfigure}\hspace{0.05\linewidth}%
  \begin{subfigure}[b]{\tempvar\linewidth}\centering
    \includegraphics[width=\linewidth]{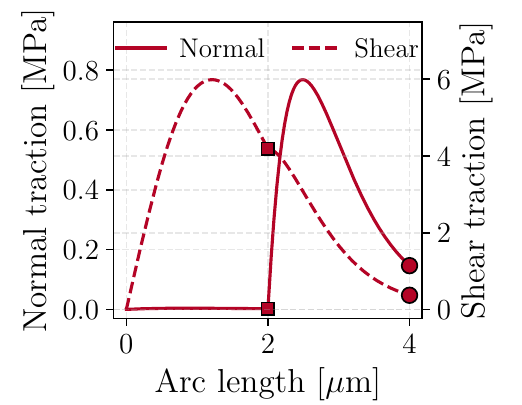}
    \caption{}\label{fig:dimitri_nonprop_tTN}
  \end{subfigure}
\caption{Non-proportional response of the model identified from CZM-Exp data
(\cref{fig:dimitri_czm-exp_mnn}). Two loading paths connect the origin to the
same final separation: N$\to$T (blue) first opens the interface in Mode~I and
then increases the tangential separation at fixed $\delta_n$, whereas
T$\to$N (red) applies the same stages in the opposite order.
(a) Loading paths in the separation plane.
(b) Cumulative mechanical dissipation (solid, left axis) and damage
(dashed, right axis).
(c,d) Normal (solid, left axis) and tangential (dashed, right axis) tractions
along the N$\to$T and T$\to$N paths, respectively. The horizontal axis in
panels (b)--(d) is the arc length along the loading path. The two paths
reach the same final separation, damage, and traction state but accumulate
different amounts of mechanical dissipation.}
  \label{fig:dimitri_nonprop}
\end{figure}

\section{Conclusions}

\label{sec:conclusions}

This work introduces a data-driven cohesive zone formulation within the
generalized standard material framework. Rather than prescribing the
traction--separation relation and unloading rules directly, the cohesive
response is defined through a learned active potential and a
mode-dependent damage resistance. The active potential is represented by an
input-convex neural network, while the inverse damage resistance is represented
by a monotone neural network. The required convexity, monotonicity,
normalization, and damage-irreversibility properties are incorporated directly
into the constitutive representation. Direct parametrization of the inverse
resistance also provides an explicit damage update and avoids a local nonlinear
inversion of the resistance relation during constitutive evaluation.

The material-point studies demonstrate that the formulation can represent a
broad range of loading, unloading, and softening responses while maintaining
irreversible damage evolution and non-negative dissipation. A single identified
model can accommodate distinct Mode~I and Mode~II responses as well as
mixed-mode behavior. When calibration data are restricted to the two pure
modes, the intermediate mixed-mode response remains underdetermined. The MNN
allows unrestricted dependence on mode mixity, whereas the full-MNN introduces
an additional inductive bias by constraining the resistance to increase with
mode mixity. When calibration data span several mode mixities, the formulation
accurately reproduces the calibration responses and predicts held-out
intermediate mixities for several classical mixed-mode cohesive laws.

Overall, the proposed formulation provides a common thermodynamic structure for learning cohesive responses, without prescribing a particular traction--separation law or phenomenological unloading relation.

\paragraph{Limitations and Future Work.}
Several modeling assumptions delimit the present scope. First, the
multiplicative degradation in \cref{eq:psi} implies that, at fixed damage, the
active unloading response is governed by the same potential
$\psi^{e}_{+}$. The present model therefore does not represent residual
separation, frictional or plastic hysteresis, damage-dependent changes in the
unloading shape, or fatigue degradation. Such behavior would require
additional internal variables or evolution mechanisms. Second, a single
scalar damage variable degrades the entire active potential and therefore
cannot represent independent degradation of the normal and tangential
responses. Extensions with multiple damage variables would permit more general
anisotropic degradation. In three dimensions, dependence on the tangential
separation only through $\lVert\boldsymbol{\delta}_t\rVert$ imposes in-plane
isotropy and does not distinguish Mode~II from Mode~III.
A further limitation concerns identifiability. Traction--separation histories
do not, in general, uniquely determine the active potential and damage
resistance individually, as illustrated in
\cref{app:mixed-mode-result}. Moreover, all identifications considered here
use synthetic, noise-free material-point data. The effects of measurement
noise, incomplete loading histories, and experimental uncertainty remain to
be investigated. 
The present formulation is restricted to small-deformation, isothermal, and
rate-independent behavior and is evaluated at the material-point level.
Natural extensions include finite-element implementation for homogeneous and
heterogeneous interfaces and inverse identification from structural-scale
measurements rather than directly observed traction--separation data.
Extensions to rate- and temperature-dependent behavior, cyclic degradation,
fatigue, and more general three-dimensional interface anisotropy are also of
interest.

\section*{Data Availability}
All data used in this work are either publicly available or can be reproduced from the information provided and are available from the corresponding author upon reasonable request.

\section*{Acknowledgment}
B.B. acknowledges support from the startup fund provided by Northwestern University. 

\section*{Declaration of generative AI and AI-assisted technologies in the manuscript preparation process}
During the preparation of this work the authors used ChatGPT and Claude in order to assist with language improvement. After using the services, the authors reviewed and edited the content as needed and take full responsibility for the content of the published article.

\bibliographystyle{plainnat}
\bibliography{bibliography}


\appendix
\crefalias{section}{appendix}

\section{Elements of Convex Analysis}
\label{app:convex}

For completeness we recall, informally, the convex-analysis notions used in the
constitutive construction; rigorous statements may be found in standard
references \cite{rockafellar1997convex}. Throughout,
$f:\R^{m}\to\R\cup\{+\infty\}$ is an extended-real-valued function and
$C\subseteq\R^{m}$ a convex set.

\paragraph{Convex function.}$f$ is convex if $f(\lambda x+(1-\lambda)y)\le\lambda f(x)+(1-\lambda)f(y)$ for all $x,y$ and $\lambda\in[0,1]$; equivalently, its epigraph $\{(x,\alpha):\alpha\ge f(x)\}$ is a convex set. Convexity underlies both the stability of the stored energy and the non-negativity of the dissipation.

\paragraph{Composition rule.}If $c_1,\dots,c_m$ are convex and
$g$ is jointly convex and non-decreasing in every argument whose $c_i$ is not
affine, then $x\mapsto g\bigl(c_1(x),\dots,c_m(x)\bigr)$ is convex
\citep[Section~3.2.4]{boyd2004convex}.

\paragraph{Lower semicontinuity.} $f$ is lower semicontinuous if $\liminf_{y\to x}f(y)\ge f(x)$ at every $x$, equivalently if all sublevel sets $\{f\le\alpha\}$ are closed. Intuitively the function may jump \emph{up} but not \emph{down} in a limit; this mild requirement ensures that minima are attained and that a potential taking the value $+\infty$ on a constraint (as in \cref{eq:phi-reduced}) is admissible.

\paragraph{Proper and closed convex function.} A convex $f$ is proper if it is finite for at least one argument and never equals $-\infty$, and closed if, in addition, it is lower semicontinuous (its epigraph is closed). These are the functions for which the notions below are well behaved.

\paragraph{Subdifferential.} For a convex $f$, the subdifferential at $x$ is the set of slopes of supporting hyperplanes,
    \begin{equation*}
      \subdiff f(x)=\bigl\{\,g:\ f(y)\ge f(x)+g\cdot(y-x)\ \ \forall y\,\bigr\}.
    \end{equation*}
It generalizes the gradient to non-smooth points: where $f$ is
differentiable it reduces to $\{\nabla f(x)\}$, whereas at a kink it is a
whole set. The evolution law $Y\in\subdiff_{\dot d}\phi$ is stated with
this notion.

\paragraph{Positive homogeneity of degree one.} $f$ is positively homogeneous of degree one if $f(cx)=c\,f(x)$ for every $c>0$.

\paragraph{Indicator function.} The indicator of $C$ is $I_C(x)=0$ for $x\in C$ and $+\infty$ otherwise. It is convex, proper, and closed exactly when $C$ is convex and closed, and encodes a constraint as an infinite energy penalty --- as with the irreversibility constraint $\dot d\ge0$.

\paragraph{Normal cone.}
For $x\in C$, the normal cone is $N_C(x)=\{\,g:\ g\cdot(y-x)\le0\ \ \forall y\in C\,\}$, the set of outward directions at $x$. It equals $\{0\}$ in the interior and widens on the boundary; it is precisely the subdifferential of the indicator function of $C$.

\section{Convexity Conditions for Composed Cohesive Potentials}
\label{app:composed-convexity}

Consider a composed function $F(\boldsymbol{x})
  =
  f\bigl(\boldsymbol{y}(\boldsymbol{x})\bigr),$
where $\boldsymbol{y}$ is the inner map and $f$ is the outer function. We
first recall the affine case.

\begin{lemma}[affine map]
\label[lemma]{lem:affine-composition}
Let $\boldsymbol{y}(\boldsymbol{x})
  =
  \boldsymbol{A}\boldsymbol{x}+\boldsymbol{b},$
and let $f$ be convex on the range of $\boldsymbol{y}$. Then
$F=f\circ\boldsymbol{y}$ is convex in $\boldsymbol{x}$. If
$\boldsymbol{A}$ is invertible, the converse also holds.
\end{lemma}

\begin{proof}
For any $\boldsymbol{x},\boldsymbol{z}$ and $\lambda\in[0,1]$,
\begin{align*}
  F\bigl(\lambda\boldsymbol{x}+(1-\lambda)\boldsymbol{z}\bigr)
  &=
  f\bigl(
    \lambda\boldsymbol{y}(\boldsymbol{x})
    +(1-\lambda)\boldsymbol{y}(\boldsymbol{z})
  \bigr)
  \\
  &\leq
  \lambda F(\boldsymbol{x})
  +(1-\lambda)F(\boldsymbol{z}),
\end{align*}
where the inequality follows from convexity of $f$. Thus, affine
precomposition preserves convexity. If $\boldsymbol{A}$ is invertible, then
\begin{equation*}
  f(\boldsymbol{y})
  =
  F\bigl(
    \boldsymbol{A}^{-1}(\boldsymbol{y}-\boldsymbol{b})
  \bigr),
\end{equation*}
so convexity of $F$ also implies convexity of $f$.
\end{proof}

Non-affine inner maps require an additional condition. We next consider the
case in which the positive-part operator acts on the first component,
\begin{equation*}
  \boldsymbol{y}(\boldsymbol{x})
  =
  \bigl(\pos{x_1},x_2,\ldots,x_m\bigr).
  \label{eq:partial-positive-map}
\end{equation*}

\begin{lemma}[positive-part component]
\label[lemma]{lem:composed-convexity}
Let $  F(\boldsymbol{x})
  =
  f\bigl(\pos{x_1},x_2,\ldots,x_m\bigr),$
where $f$ is defined on
$[0,\infty)\times\R^{m-1}$. Then $F$ is convex in
$\boldsymbol{x}$ if and only if $f$ is jointly convex and
non-decreasing in its first argument.
\end{lemma}

\begin{proof}
Suppose first that $f$ is jointly convex and non-decreasing in its first
argument. Since
\begin{equation*}
  \pos{\lambda x_1+(1-\lambda)z_1}
  \leq
  \lambda\pos{x_1}+(1-\lambda)\pos{z_1},
\end{equation*}
while the remaining components are affine, monotonicity followed by
convexity of $f$ gives
\begin{equation*}
  F\bigl(\lambda\boldsymbol{x}+(1-\lambda)\boldsymbol{z}\bigr)
  \leq
  \lambda F(\boldsymbol{x})+(1-\lambda)F(\boldsymbol{z}).
\end{equation*}

Conversely, assume that $F$ is convex. On the half-space $x_1\geq0$, the
inner map is the identity, so $f$ is convex on its domain. Fixing
$(x_2,\ldots,x_m)$ gives the convex function
\begin{equation*}
  h(t)
  =
  f\bigl(\pos{t},x_2,\ldots,x_m\bigr).
\end{equation*}
Since $h$ is constant for $t\leq0$, convexity implies that it is
non-decreasing for $t\geq0$. Hence, $f$ is non-decreasing in its first
argument.
\end{proof}

The same argument applies when the positive-part operator acts on several
components. In that case, the outer function must be jointly convex and
non-decreasing in each rectified argument.

A second case relevant to cohesive laws arises when one component enters
through its magnitude.

\begin{lemma}[positive part and a magnitude]
\label[lemma]{lem:positive-magnitude-composition}
Let $  F(\boldsymbol{x})
  =
  f\bigl(\pos{x_1},\lvert x_2\rvert,x_3,\ldots,x_m\bigr).$
Then $F$ is convex in $\boldsymbol{x}$ if and only if $f$ is jointly
convex and non-decreasing in its first two arguments.
\end{lemma}

\begin{proof}
Sufficiency follows from the convexity of $\pos{x_1}$ and
$\lvert x_2\rvert$, together with joint convexity and monotonicity of $f$
in the corresponding arguments.

Conversely, restrict to $x_1\geq0$ and $x_2\geq0$, where the inner map is
the identity. This gives joint convexity of $f$. Monotonicity in the first
argument follows from \cref{lem:composed-convexity}. For fixed values of
the remaining arguments,
\begin{equation*}
  t
  \mapsto
  f\bigl(x_1,\lvert t\rvert,x_3,\ldots,x_m\bigr)
\end{equation*}
is convex and even. It is therefore non-decreasing for $t\geq0$, which
gives monotonicity of $f$ in its second argument.
\end{proof}

We now specialize these results to the interface separation. The coordinate
map $  \boldsymbol{\delta}
  \mapsto
  (\delta_n,\delta_t)$
is affine and therefore does not introduce any additional condition.

If the two sliding directions are distinguished, the active energy has the
form $  \psi^{e}_{+}(\boldsymbol{\delta})
  =
  \widehat{\psi}^{e}_{+}
  \bigl(\pos{\delta_n},\delta_t\bigr).$
By \cref{lem:composed-convexity}, it is convex in
$\boldsymbol{\delta}$ if and only if
$\widehat{\psi}^{e}_{+}$ is jointly convex and non-decreasing in its first
argument.

More commonly, symmetry with respect to the sliding direction is imposed: $  \psi^{e}_{+}(\boldsymbol{\delta})
  =
  \widehat{\psi}^{e}_{+}
  \bigl(\pos{\delta_n},\lvert\delta_t\rvert\bigr).$
In this case, \cref{lem:positive-magnitude-composition} requires
$\widehat{\psi}^{e}_{+}$ to be jointly convex and non-decreasing in both
arguments.

Similarly, $  \psi^{e}_{-}(\boldsymbol{\delta})
  =
  \widehat{\psi}^{e}_{-}\bigl(\pos{-\delta_n}\bigr)$
is convex in $\boldsymbol{\delta}$ if and only if
$\widehat{\psi}^{e}_{-}$ is convex and non-decreasing.

\section{Hyperparameter Settings and Implementation Details}
\label{app:hyperparams}
All neural-network models and training procedures are implemented in Python using PyTorch \citep{paszke2019pytorch}. 

The active potential is represented by an ICNN with softplus-based
activations. 
To maintain a non-vanishing activation slope, we use
$\sigma(h)=\operatorname{softplus}(h)+10^{-2}h$. The added linear term preserves convexity and monotonicity while reducing near-saturation of the activation. The separation inputs and energy output of the ICNN are nondimensionalized
using characteristic scales so that their numerical magnitudes remain of order
one. The separation is normalized by the separation corresponding to the peak
traction, while the energy is normalized by the area under the corresponding
traction--separation curve.

The inverse resistance is represented by the MNN.
In mixed-mode problems, the confluence and nonlinear streams use the same widths. The
full-MNN is considered only in the pure-mode interpolation study of \cref{sec:examp-MNN-MRQS}; in that case, the unrestricted streams are omitted and the architecture is specified by the exponential-stream widths alone. The driving force $Y$ is passed directly to the MNN. A small linear term
$\varepsilon Y$ is added to the MNN output to prevent the slope from becoming
arbitrarily small when the bounded activation functions approach saturation.
We use $\varepsilon=10^{-4}$.

The restriction of the learned inverse resistance to the admissible damage
range is imposed pointwise at each constitutive update. To improve numerical
smoothness near damage onset during training, we replace the positive-part
operator by the scaled softplus approximation
\begin{equation}
  d_{\mathrm{tr}}
  =
  \min\left(
    \langle \hat R(Y,\beta)\rangle_{\kappa},
    1
  \right),
  \qquad
  \langle x\rangle_{\kappa}
  \coloneqq
  {\kappa}\log\left(1+\exp(x/\kappa)\right),
  \label{eq:soft_clamp}
\end{equation}
where $\kappa>0$ controls the width of the smoothed transition. The onset corner is
therefore regularized, while the upper bound remains a hard saturation so that
$d_{\mathrm{tr}}\le1$ and complete damage is attained exactly. As $\kappa\to0$,
$\langle x\rangle_{\kappa}\to\pos{x}$, and \cref{eq:soft_clamp} recovers the hard
clamp in the limiting case.

All models are trained using Adam \citep{diederik2014adam} with an initial
learning rate of $3\times10^{-3}$ and cosine annealing with warm restarts
\citep{loshchilov2016sgdr}. Unless otherwise stated, the learning rate is
annealed to $\eta_{\min}=10^{-5}$ within each cycle. The first cycle spans
one fifth of the total epoch budget, and the cycle length is doubled after
each restart. For the examples in \cref{sec:examp-dimitri}, we instead use a
single restart after the first fifth of the epoch budget. The second cycle
then spans the remaining epochs and anneals to
$\eta_{\min}=3\times10^{-4}$. Gradient norms are clipped at $50$
\citep{pascanu2013difficulty}.

Problem-specific details are provided below. 

\paragraph{\cref{sec:ex-mode-I}} 
The study of \cref{fig:flex} uses an ICNN of three
hidden layers of eight units each and an
MNN of two
hidden layers of four units each; the unload--reload study of
\cref{fig:dissipation_multiunload} uses a larger ICNN of three
hidden layers of sixteen units each and an
MNN of two
hidden layers of eight units each, since both the loading and the unloading responses are more complex. Both are trained for $40{,}000$ epochs with  $\kappa=0.01$.

\paragraph{\cref{sec:examp-MNN-MRQS}}
The ICNN is held at three
hidden layers of eight units each in every case, so the comparison isolates the
inverse-resistance representation. 
Three resistance capacities are considered for each architecture --- two hidden
layers of two, four, or eight units --- giving $42$, $114$, and $354$ parameters for the MNN and $16$, $38$, and $106$ for the full-MNN.
The confluence and nonlinear streams use the same widths.
Each configuration is trained for $30{,}000$ epochs with
$\kappa=0.03$. An $L_2$ regularization penalty with coefficient $10^{-5}$ is applied to the MNN weights and biases to reduce overfitting. 

\paragraph{\cref{sec:examp-dimitri}}
The ICNN and MNN each use two hidden layers with eight units per layer. Training is performed for $50{,}000$ epochs with $\kappa=0.03$.

No systematic hyperparameter search or architecture optimization was performed; the network sizes were selected empirically to provide sufficient representational capacity for each example. The reported configurations should therefore not be interpreted as optimal architectures. 

\section{Additional Mixed-Mode Results}
\label{app:mixed-mode-result}

This appendix provides additional results for the mixed-mode example of
\cref{sec:examp-MNN-MRQS}. We report the optimization histories for the
different MNN and full-MNN capacities and compare the active potentials and
damage resistances identified by the highest-capacity models.

\Cref{fig:mnn_mrqs_mm_loss} shows the training loss for the three network
capacities considered for each inverse-resistance representation, with five
random initializations per configuration. The optimization exhibits greater
sensitivity to initialization for the smaller networks, while the
highest-capacity models attain the lowest terminal losses. The transient
increases in the loss coincide with the learning-rate restarts described in
\cref{app:hyperparams}.

\begin{figure}[htbp]
     \centering
     \begin{subfigure}[b]{\textwidth}
         \centering
         \includegraphics[width=0.32\linewidth]{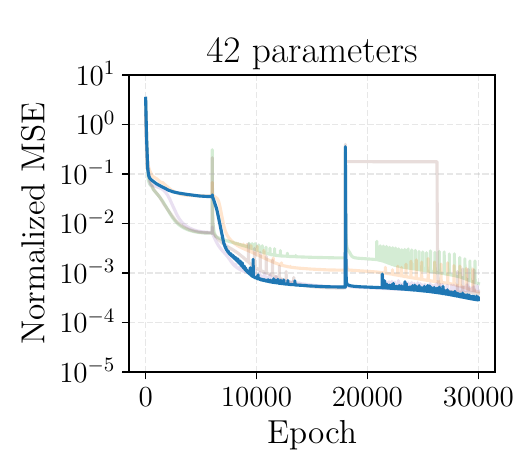}\hfill
         \includegraphics[width=0.32\linewidth]{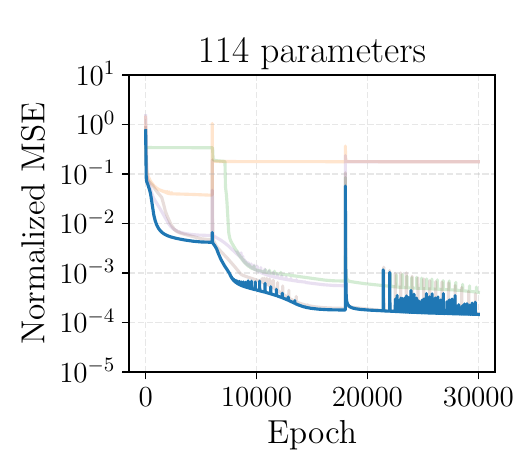}\hfill
         \includegraphics[width=0.32\linewidth]{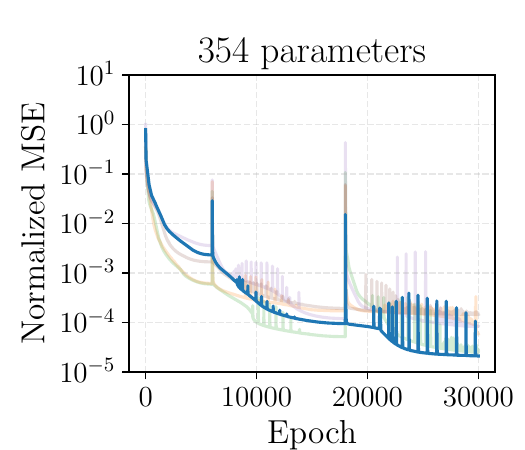}
         \caption{MNN}
         \label{fig:mnn_mrqs_mm_loss_mnn}
     \end{subfigure}\\[1ex]
     \begin{subfigure}[b]{\textwidth}
         \centering
         \includegraphics[width=0.32\linewidth]{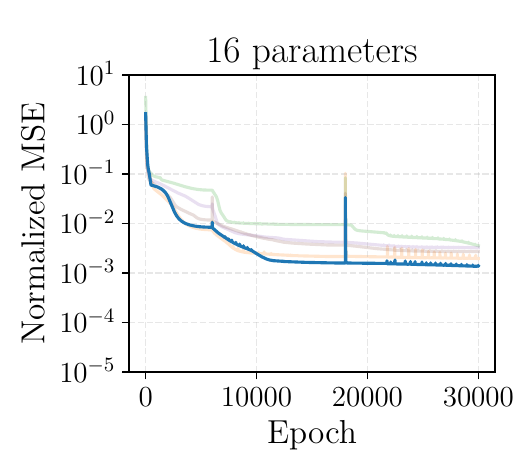}\hfill
         \includegraphics[width=0.32\linewidth]{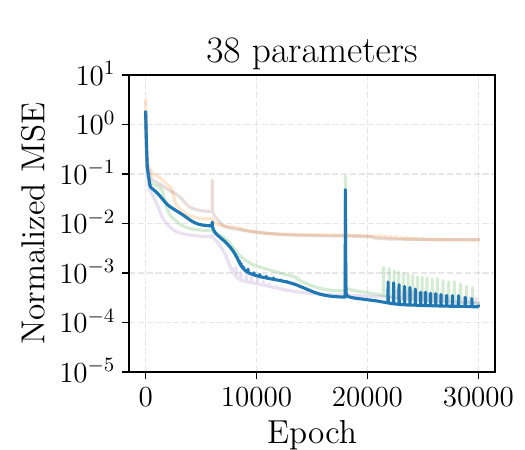}\hfill
         \includegraphics[width=0.32\linewidth]{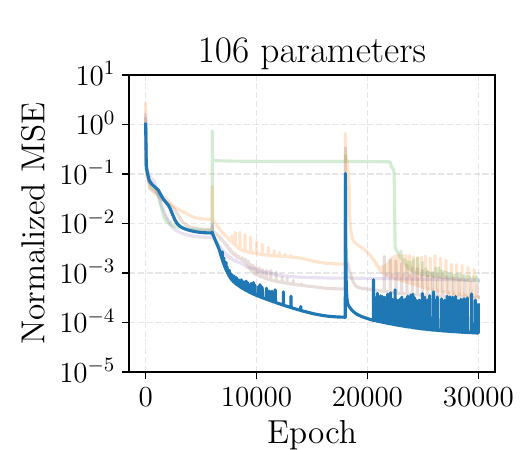}
         \caption{full-MNN}
         \label{fig:mnn_mrqs_mm_loss_fmnn}
     \end{subfigure}
     \caption{Training-loss histories for (a) the MNN and (b) the full-MNN. Three
     resistance-network capacities are considered, with five random
     initializations for each configuration. Panel titles give the number of
     resistance-network parameters. The solid blue curve denotes the
     lowest-loss realization and the faint curves denote the remaining
     initializations.}
     \label{fig:mnn_mrqs_mm_loss}
\end{figure}

Although the MNN and full-MNN reproduce the observed pure-mode
traction--separation responses with comparable accuracy, the corresponding
latent constitutive functions need not be unique.
\Cref{fig:mnn_mrqs_mm_psi,fig:mnn_mrqs_mm_R} show the active potential and
damage resistance identified by the highest-capacity model of each
architecture. The two models recover noticeably different
$\psi^{e}$ and $r$ while producing similar responses on the calibration paths.
This illustrates that traction--separation data constrain the observable
constitutive response more directly than the individual latent functions.

The difference between the resistance surfaces also reflects the distinct
assumptions imposed by the two architectures. The MNN leaves the dependence on
$\beta$ unrestricted, whereas the full-MNN constrains the resistance to
increase with mode mixity. Consequently, the full-MNN produces the monotonic
fracture-energy trend reported in \cref{fig:mixed_mode_beta_sweep_ts}, while
the unconstrained MNN admits the non-monotonic interpolation observed there.

\begin{figure}[htbp]
  \centering
  \begin{subfigure}[b]{0.35\linewidth}
    \centering
    \includegraphics[width=\linewidth]{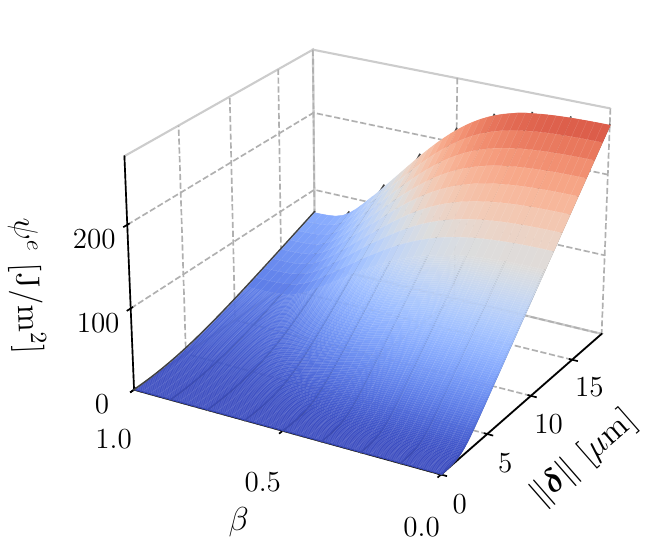}
    \caption{MNN}
    \label{fig:beta_surf_psi_mnn}
  \end{subfigure}%
  \hspace{0.04\linewidth}%
  \begin{subfigure}[b]{0.35\linewidth}
    \centering
    \includegraphics[width=\linewidth]{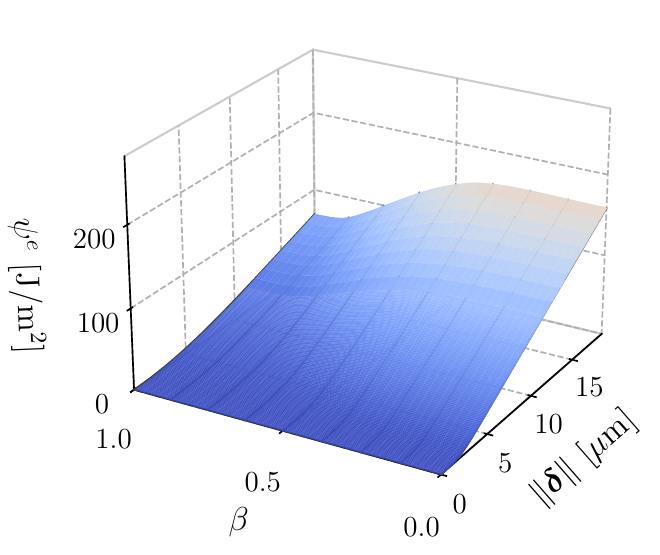}
    \caption{full-MNN}
    \label{fig:beta_surf_psi_fmnn}
  \end{subfigure}
  \caption{Active potential $\psi^{e}$ identified by (a) the MNN and
  (b) the full-MNN models, evaluated along proportional loading directions as
  a function of separation magnitude $\lVert\boldsymbol{\delta}\rVert$ and
  mode mixity $\beta$.}
  \label{fig:mnn_mrqs_mm_psi}
\end{figure}

\begin{figure}[htbp]
  \centering
  \begin{subfigure}[b]{0.35\linewidth}
    \centering
    \includegraphics[width=\linewidth]{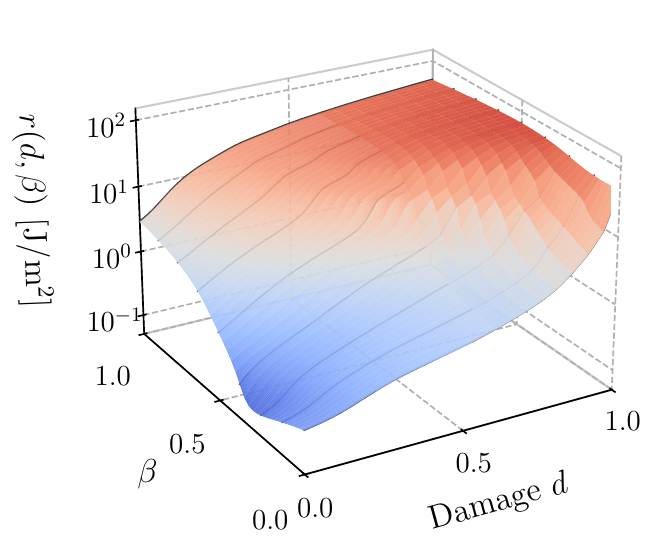}
    \caption{MNN}
    \label{fig:beta_surf_r_mnn}
  \end{subfigure}%
  \hspace{0.04\linewidth}%
  \begin{subfigure}[b]{0.35\linewidth}
    \centering
    \includegraphics[width=\linewidth]{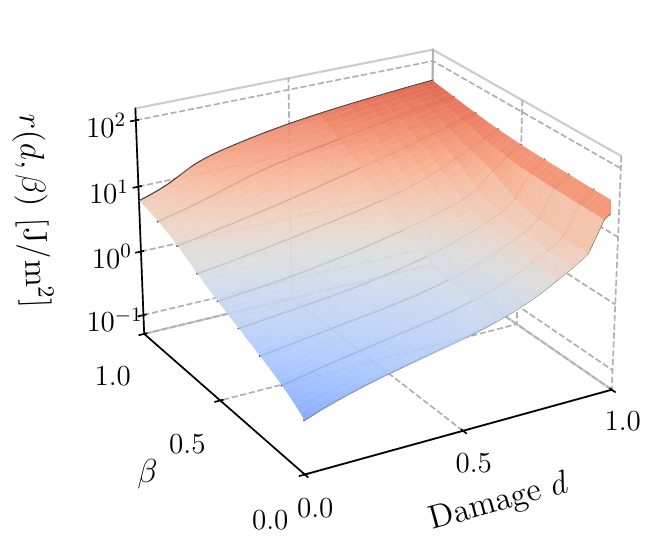}
    \caption{full-MNN}
    \label{fig:beta_surf_r_fmnn}
  \end{subfigure}
  \caption{Damage resistance $r(d,\beta)$ identified by (a) the MNN and
  (b) the full-MNN models.}
  \label{fig:mnn_mrqs_mm_R}
\end{figure}

\section{Reference Cohesive Laws for the Classical-Model Study}
\label{app:dimitri_laws}

The three reference cohesive laws used in \cref{sec:examp-dimitri} are
summarized below. We consider the modified exponential law of
\citet{van2006improved}, a bilinear damage law based on
\citet{camanho2003numerical}, and the PPR potential of \citet{Park2009}.
The parameters are chosen as
\begin{equation}
  T_n = 3~\mathrm{MPa},
  \qquad
  T_t = 6~\mathrm{MPa},
  \qquad
  G_{Ic} = 4~\mathrm{J/m^2}.
\end{equation}
For the exponential and bilinear laws, we set
$G_{IIc}=12~\mathrm{J/m^2}$, while the PPR law uses
$G_{IIc}=6~\mathrm{J/m^2}$. Synthetic data are generated along proportional
loading paths with fixed mode mixity
$\beta=\lvert\delta_t\rvert/(\pos{\delta_n}+\lvert\delta_t\rvert)$, as defined
in \cref{eq:mode-mixity}. Only positive normal and tangential separations are
used in the reference loading paths.

\begin{remark}
In the following, some symbols are reused locally to avoid an unnecessary proliferation of notation; their meaning will be clear from context.
\end{remark}

\paragraph{Exponential law (CZM-Exp).}
We use the modified exponential mixed-mode law of
\citet{van2006improved},
\begin{align}
  t_n
  &=
  T_n\,e\,x\,e^{-x}e^{-y^2},
  \\
  t_t
  &=
  T_t\,\sqrt{2e}\,y(1+x)e^{-x}e^{-y^2},
\end{align}
where
\begin{equation}
  x=\frac{\delta_n}{\delta_n^0},
  \qquad
  y=\frac{\delta_t}{\delta_t^0},
\end{equation}
and
\begin{equation}
  \delta_n^0
  =
  \frac{G_{Ic}}{eT_n},
  \qquad
  \delta_t^0
  =
  \frac{G_{IIc}}{\sqrt{e/2}\,T_t},
\end{equation}
with $e=\exp(1)$. 
Here $x$ and $y$ denote the normalized normal and tangential separations, respectively.
These definitions give the prescribed pure-mode peak
tractions $T_n$ and $T_t$ and fracture energies $G_{Ic}$ and $G_{IIc}$.

\paragraph{Bilinear law (CZM-BL).}
We use a bilinear damage law based on \citet{camanho2003numerical}, combined
here with a power-law mixed-mode fracture criterion. A common penalty
stiffness
\begin{equation}
  K=10^{13}~\mathrm{Pa/m}
\end{equation}
is used in the normal and tangential directions. The pure-mode damage-onset
separations are
\begin{equation}
  \delta_n^0=\frac{T_n}{K},
  \qquad
  \delta_t^0=\frac{T_t}{K}.
\end{equation}
Defining the effective separation and loading-direction ratio as
\begin{equation}
  \delta_m
  =
  \sqrt{\pos{\delta_n}^{\,2}+\delta_t^2},
  \qquad
  \rho
  =
  \frac{\delta_t}{\pos{\delta_n}},
\end{equation}
the mixed-mode onset separation is
\begin{equation}
  \delta_m^0
  =
  \delta_n^0\delta_t^0
  \sqrt{
    \frac{1+\rho^2}
    {(\delta_t^0)^2+\rho^2(\delta_n^0)^2}
  }.
\end{equation}
Its pure Mode~II limit is $\delta_m^0=\delta_t^0$.

The mixed-mode fracture energy is prescribed by
\begin{equation}
  G_c(\beta)
  =
  \left[
    \left(
      \frac{\cos^2\theta}{G_{Ic}}
    \right)^{\gamma}
    +
    \left(
      \frac{\sin^2\theta}{G_{IIc}}
    \right)^{\gamma}
  \right]^{-1/\gamma},
  \qquad
  \gamma=1,
\end{equation}
where
\begin{equation}
  \cos^2\theta
  =
  \frac{\pos{\delta_n}^{\,2}}{\delta_m^2},
  \qquad
  \sin^2\theta
  =
  \frac{\delta_t^2}{\delta_m^2}.
\end{equation}
The corresponding complete-failure separation is
\begin{equation}
  \delta_m^f
  =
  \frac{2G_c(\beta)}{K\delta_m^0}.
\end{equation}

Let
\begin{equation}
  \bar{\delta}_m(t)
  =
  \max_{\tau\le t}\delta_m(\tau)
\end{equation}
denote the maximum effective separation attained along the loading history.
The damage variable is
\begin{equation}
  d
  =
  \min\left[
    \max\left(
      \frac{
        \delta_m^f(\bar{\delta}_m-\delta_m^0)
      }{
        \bar{\delta}_m(\delta_m^f-\delta_m^0)
      },
      0
    \right),
    1
  \right].
\end{equation}
The tractions are then
\begin{equation}
  t_n
  =
  (1-d)K\pos{\delta_n},
  \qquad
  t_t
  =
  (1-d)K\delta_t.
\end{equation}

\paragraph{Park--Paulino--Roesler law (PPR).}
The PPR potential of \citet{Park2009} is used with the shape parameters
$\alpha=5$ and $\beta_{\mathrm p}=1.3$, where $\beta_{\mathrm p}$ denotes
the tangential PPR exponent and is distinct from the mode-mixity variable
$\beta$. The initial-slope parameters are
$\lambda_n=0.15$ and $\lambda_t=0.30$. The corresponding exponents are
\begin{equation}
  m
  =
  \frac{
    \alpha(\alpha-1)\lambda_n^2
  }{
    1-\alpha\lambda_n^2
  },
  \qquad
  n
  =
  \frac{
    \beta_{\mathrm p}(\beta_{\mathrm p}-1)\lambda_t^2
  }{
    1-\beta_{\mathrm p}\lambda_t^2
  }.
\end{equation}
Following \citet{Park2009}, $m$ and $n$ denote the PPR shape exponents.

The energy constants are
\begin{equation}
  \Gamma_n
  =
  \left(-G_{Ic}\right)^{
    \pos{G_{Ic}-G_{IIc}}/(G_{Ic}-G_{IIc})
  }
  \left(\frac{\alpha}{m}\right)^m,
\end{equation}
and
\begin{equation}
  \Gamma_t
  =
  \left(-G_{IIc}\right)^{
    \pos{G_{IIc}-G_{Ic}}/(G_{IIc}-G_{Ic})
  }
  \left(\frac{\beta_{\mathrm p}}{n}\right)^n.
\end{equation}
For the present choice $G_{IIc}>G_{Ic}$, these reduce to
\begin{equation}
  \Gamma_n
  =
  \left(\frac{\alpha}{m}\right)^m,
  \qquad
  \Gamma_t
  =
  -G_{IIc}
  \left(\frac{\beta_{\mathrm p}}{n}\right)^n.
\end{equation}

The final normal and tangential separations are determined from the prescribed
pure-mode cohesive strengths,
\begin{align}
  \delta_n^f
  &=
  \frac{G_{Ic}}{T_n}
  \alpha\lambda_n(1-\lambda_n)^{\alpha-1}
  \left(
    \frac{\alpha}{m}+1
  \right)
  \left(
    \frac{\alpha}{m}\lambda_n+1
  \right)^{m-1},
  \\
  \delta_t^f
  &=
  \frac{G_{IIc}}{T_t}
  \beta_{\mathrm p}\lambda_t
  (1-\lambda_t)^{\beta_{\mathrm p}-1}
  \left(
    \frac{\beta_{\mathrm p}}{n}+1
  \right)
  \left(
    \frac{\beta_{\mathrm p}}{n}\lambda_t+1
  \right)^{n-1}.
\end{align}
Introducing
\begin{equation}
  \hat{\delta}_n
  =
  \frac{\delta_n}{\delta_n^f},
  \qquad
  \hat{\delta}_t
  =
  \frac{\delta_t}{\delta_t^f},
\end{equation}
the normal traction is
\begin{align}
  t_n
  &=
  \frac{\Gamma_n}{\delta_n^f}
  \left[
    m(1-\hat{\delta}_n)^{\alpha}
    \left(
      \frac{m}{\alpha}+\hat{\delta}_n
    \right)^{m-1}
    -
    \alpha(1-\hat{\delta}_n)^{\alpha-1}
    \left(
      \frac{m}{\alpha}+\hat{\delta}_n
    \right)^m
  \right]
  \notag\\
  &\quad\times
  \left[
    \Gamma_t
    (1-\hat{\delta}_t)^{\beta_{\mathrm p}}
    \left(
      \frac{n}{\beta_{\mathrm p}}+\hat{\delta}_t
    \right)^n
    +
    \pos{G_{IIc}-G_{Ic}}
  \right],
\end{align}
and the tangential traction is
\begin{align}
  t_t
  &=
  \frac{\Gamma_t}{\delta_t^f}
  \left[
    n(1-\hat{\delta}_t)^{\beta_{\mathrm p}}
    \left(
      \frac{n}{\beta_{\mathrm p}}+\hat{\delta}_t
    \right)^{n-1}
    -
    \beta_{\mathrm p}
    (1-\hat{\delta}_t)^{\beta_{\mathrm p}-1}
    \left(
      \frac{n}{\beta_{\mathrm p}}+\hat{\delta}_t
    \right)^n
  \right]
  \notag\\
  &\quad\times
  \left[
    \Gamma_n
    (1-\hat{\delta}_n)^{\alpha}
    \left(
      \frac{m}{\alpha}+\hat{\delta}_n
    \right)^m
    +
    \pos{G_{Ic}-G_{IIc}}
  \right].
\end{align}
These expressions are evaluated over
$0\le\delta_n\le\delta_n^f$ and
$0\le\delta_t\le\delta_t^f$ for the proportional loading paths considered
here. The corresponding traction component is set to zero once its terminal
separation is exceeded.

\end{document}